\documentclass[reqno]{amsart}
\usepackage{amssymb}

\DeclareMathOperator{\supp}{supp}
\DeclareMathOperator{\tr}{tr}
\DeclareMathOperator{\Ran}{Ran}
\DeclareMathOperator{\dist}{dist}
\DeclareMathOperator{\Ker}{Ker}

\newcommand{\eq}[1]{\eqref{#1}}
\newcommand{\up}[1]{^{(#1)}}
\newcommand{\bp}[1]{^{\{#1\}}}

\newcommand\beq{\begin{equation}}
\newcommand\eeq{\end{equation}}

\newtheorem{theorem}{Theorem}[section]

\newtheorem{lemma}[theorem]{Lemma}
\newtheorem{corollary}[theorem]{Corollary}

\numberwithin{equation}{section}

\newcommand\R{\mathbb R}
\newcommand\N{\mathbb N}

\newcommand\Z{\mathbb Z}

\newcommand\di{\mathrm d}

\newcommand{\T}{\mathcal{T}}

\newcommand\e{\mathrm{e}}

\renewcommand\P{\mathbb P}
\newcommand\E{\mathbb E}

\newcommand{\qnorm}[1]{\left\lVert\!\left\Vert #1 \right\rVert\!\right\rVert}
\newcommand{\abs}[1]{\left\lvert #1 \right\rvert}
\newcommand{\norm}[1]{\left\lVert #1 \right\rVert}
\newcommand{\set}[1]{\left\{ #1 \right\}}
\newcommand{\pa}[1]{\left( #1 \right)}
\newcommand{\pb}[1]{\left[ #1 \right]}

\begin{document}

\title[Ergodic Landau Hamiltonians]
{Quantization of the Hall conductance and delocalization in ergodic Landau Hamiltonians}

\author[F. Germinet]{Fran\c cois Germinet}
\address[Germinet]{ Universit\'e de Cergy-Pontoise,
D\'epartement de Math\'ematiques, CNRS UMR 8088, IUF,
95000 Cergy-Pontoise, France}
\email{germinet@math.u-cergy.fr}

\author[A. Klein]{Abel Klein}
\address[Klein]{University of California, Irvine,
Department of Mathematics,
Irvine, CA 92697-3875,  USA}
\email{aklein@uci.edu}

\author[J. H. Schenker]{Jeffrey H. Schenker}
\address[Schenker]{Michigan State University, Department of Mathematics, East Lansing, MI 48823, USA}
\email{jeffrey@math.msu.edu}


\thanks{A.K was  supported in part by NSF Grant DMS-0457474.}

%
\begin{abstract} 
We prove quantization of the Hall conductance  for continuous ergodic Landau Hamiltonians under a condition on the decay of the Fermi projections.   This condition and continuity of the integrated density of states are shown to imply continuity of the Hall conductance. In addition, we 
prove the existence of  delocalization near each Landau level  for these two-dimensional
Hamiltonians.  More precisely, we prove that for some  ergodic Landau Hamiltonians there exists  an energy $E$ near each Landau level where a ``localization length'' diverges.
  For  the Anderson-Landau Hamiltonian we also obtain a transition between dynamical localization and dynamical delocalization in the Landau bands, with a minimal rate of transport,  even  in cases when the spectral gaps are closed. 
\end{abstract}

\maketitle
\setcounter{tocdepth}{1}
\tableofcontents

\section{Introduction}

Ergodic Landau Hamiltonians   
describe an  electron moving in  a very thin flat conductor
with impurities under the influence of a constant
magnetic field perpendicular to the plane of the conductor. 
They play an important role in the understanding of the quantum Hall effect
\cite{L,AA,T,Ha,NT,Ku,Be,ASS,BES}.  Laughlin's argument relies on  
the assumption that under weak disorder and strong magnetic field the energy
spectrum consists of
bands of extended states separated by energy regions of localized states and/or energy
gaps \cite{L,Ha,AA,T}. 
 Kunz \cite{Ku} formulated assumptions under which he
derived the divergence of a  ``localization length" 
near each Landau level at weak disorder. 

Previous to our recent paper \cite{GKS}, there had been no rigorous results concerning delocalization for continuous ergodic Landau Hamiltonians.   Divergence of a  ``localization length" had only been  proved  for an ergodic Landau Hamiltonian in  a tight-binding approximation, a discrete ergodic Schr\"odinger operator. 
The first  results were obtained by
Bellissard, van Elst and Schulz-Baldes \cite{BES}, who proved that, for an
ergodic Landau Hamiltonian in  a tight-binding approximation,
if the Hall
conductance jumps from one integer value to another between two Fermi
energies, then there is an energy between these Fermi energies at which a
certain
localization length diverges.  Their results relied on a proof of
the quantization of Hall conductance (the quantum Hall effect) for ergodic Landau Hamiltonians in a tight binding representation (discrete ergodic Landau Hamiltonians) in energy intervals characterized by  a condition on the decay of the Fermi projections. Their proof relies on  noncommutative geometry and the Dixmier trace. Aizenman and Graf \cite{AG} gave a more elementary derivation of this result, incorporating  ideas of Avron,  Seiler and Simon \cite{ASS}, paying the price of  a slightly stronger condition on the decay of the Fermi projections.

In \cite{GKS} we proved that the (continuous)   Anderson-Landau Hamiltonian (the random Landau Hamiltonian  in \cite{GKS}) exhibits dynamical delocalization in each Landau band.  More precisely, under the disjoint bands condition (open spectral gaps between Landau bands), which holds (bounded potentials)
under  weak disorder and/or strong magnetic field, we proved the existence of a transition between   dynamical localization and dynamical delocalization in each Landau band, with a lower bound on the rate of transport.  
We used nontrivial consequences  of the multiscale analysis for
random Schr\"odinger operators  to prove that the Hall conductance for  the  Anderson-Landau Hamiltonian  is well defined and  
constant in intervals of  dynamical localization. We used  the knowledge of the  precise values
of the Hall conductance for the (free) Landau Hamiltonian: it is constant
between Landau levels and jumps by one at each Landau level,   a well
known fact (e.g., \cite{ASS,BES}).  In addition, we showed that the Hall conductance is constant as a function of the disorder
parameter in the gaps between the  Landau bands,
a result previously  derived by Elgart and Schlein \cite{ES}  for smooth potentials. Under the disjoint bands conditions (open spectral gaps), we combined these ingredients to conclude that there must be
dynamical delocalization as we cross a  Landau band. Moreover, since the existence of dynamical localization at the edges of these Landau bands was
known \cite{CH,Wa,GKgafa}, we proved  the existence of dynamical mobility
edges. 

In \cite{GKS} we circumvented the use of  the quantization of the Hall conductance.
For continuous Landau Hamiltonians quantization of the   Hall conductance had only been known  on spectral gaps \cite{ASS}.
A proof of  quantization of the   Hall conductance inside the spectrum of continuous ergodic Landau Hamiltonians has been a long-standing open problem. Although it was  promised in 1994 \cite{BES}, the proof never appeared.   (As mentioned in \cite{BES}, in the discrete case their proof studies a compact noncommutative manifold, while in the continuous case the corresponding noncommutative manifold is locally compact, but not compact.)

In this article we prove quantization of the Hall conductance
for continuous ergodic Landau Hamiltonians under a condition on the decay of the Fermi projections.  We also show that this condition and  continuity of the integrated density of states imply continuity of the Hall conductance. In particular, we get quantization and continuity of the Hall conductance
for the Anderson-Landau Hamiltonian in the region of localization.

Our condition on the decay of the Fermi projections is reminiscent of the condition used in \cite{AG}, but it is not the same because of differences between the continuous and the discrete cases.  Although the weaker condition given in \cite{BES} is very natural (it  was shown by
Bouclet and the authors \cite{BGKS} to be sufficient for a rigorous
derivation of the Kubo-St\u{r}eda formula for the Hall conductance
in continuous  ergodic Landau Hamiltonians), its use for a derivation of the quantization of the Hall conductance
seems to require methods of noncommutative geometry and the Dixmier trace that have not been extended to the continuous case.

In \cite{GKS} we did not use the quantization of  the Hall conductance, but required
 the disjoint bands condition.  The results in this paper not only give a new proof of the delocalization results in \cite{GKS}, but they allow the extension of those results to ergodic Landau Hamiltonians, in the sense of divergence of a ``localization length".

In this paper  we  go beyond the disjoint bands condition, proving  dynamical delocalization  in the Landau bands  for  the Anderson-Landau Hamiltonian
in cases where the spectral gaps are closed.  Using our results on the quantization of the Hall conductance,  we prove the existence of a transition between   dynamical localization and dynamical delocalization in  a Landau band, with a lower bound on the rate of transport, for Anderson-Landau Hamiltonians with closed  spectral gaps.  Although in this paper we assume, as in \cite{GKS}, that the  potentials are bounded, this restriction can be removed.
This extension appears in a companion article \cite{GKM}, which considers an Anderson-Landau Hamiltonian with unbounded random amplitudes (e.g.,  with a Gaussian distribution), where all the gaps close as soon as  the disorder  is turned on.  The main results of this paper still hold for such unbounded  Anderson-Landau Hamiltonians;  the theorem concerning the existence of a dynamical transition is stated below  for completeness.

\section{Definitions and main results}
\label{sectintro}

We consider a 
$\mathbb{Z}^2$-ergodic Landau Hamiltonian 
\begin{equation} \label{landauh} 
H_{B,\lambda,\omega} =H_B +
\lambda  V_\omega \quad \mathrm{on} \quad
\mathrm{L}^2(\mathbb{R}^2, {\mathrm{d}}x), 
\end{equation}
where $H_B$ is the (free) Landau Hamiltonian,
\begin{equation}\label{free Landau}
H_B =  (-i\nabla-\mathbf{A})^2 \quad \text{with} \quad 
\mathbf{A}= \frac B2 (x_2,-x_1)
\end{equation} 
($\mathbf{A}$ is the vector potential and
$B>0$ is the strength of the magnetic field, we use the symmetric gauge
and incorporated the charge of the electron in the vector potential), 
$\lambda \ge 0$ is
the disorder parameter,   and
$V_\omega$ is a bounded ergodic (real) potential.  Thus, there is a probability space $(\Omega, \P)$ equipped with an ergodic group $\{\tau(a); \ a \in \Z^2\}$ of
measure preserving transformations, a potential-valued map $V_\omega$ on $\Omega$, measurable in the sense that $\langle \phi, V_\omega \phi \rangle$ is a measurable function of $\omega$ for all $\phi \in C^\infty_c(\mathbb{R}^2)$.   Such a family of potentials includes random as well as quasiperiodic potentials.  We  assume that
\beq \label{Vomega}
-M_1 \le   V_\omega(x) \le M_2, \quad\text{where} \quad M_1,M_2 \in [0,\infty)\quad \text{with} \quad M_1 + M_2 >0,
\eeq
and
\beq
V_\omega(x-a)=   V_{\tau_a\omega}(x) \quad \text{for all  $a \in \Z^2$}.
\eeq

An important example of an ergodic Landau Hamiltonian is the Anderson-Landau Hamiltonian  
\beq \label{ALH}
H_{B,\lambda,\omega}\up{A}:= H_B + \lambda V_\omega\up{A},
\eeq 
where $ V_\omega\up{A}$ is the random potential 
\begin{equation} \label{potVL}
 V_\omega\up{A}(x) = \sum_{i \in\mathbb{Z}^2} \omega_i\, u(x-i)  ,
\end{equation} 
with $u(x) \ge 0$  a bounded measurable function with compact support, 
$u(x) \ge u_0$ on some nonempty open set for some constant $u_0 >0$, and 
$\omega =\{\omega_i; \ i\in\mathbb{Z}^2\}$ a 
family of independent,
identically distributed random variables taking values in a bounded interval 
$[-M_1, M_2]$ ($0 \le M_1,M_2 < \infty$, $ M_1 + M_2 > 0$), whose common probability
distribution $\mu$ has a bounded density $\rho$.  Without
loss of generality
we set $\left\|\sum_{i \in\mathbb{Z}^2} \, u(x-i)\right\|_\infty=1$, and hence
 $\, -M_1 \le   V_\omega\up{A}(x)\le M_2$.

An ergodic Landau Hamiltonian $ H_{B,\lambda,\omega}$ is a self-adjoint measurable
operator, i.e., with probability one  $ H_{B,\lambda,\omega}$ is a self-adjoint operator and the mappings $\omega \to f( H_{B,\lambda,\omega})$ are strongly
measurable for all bounded measurable functions on $\mathbb{R}$ (cf. \cite{PF}).   The magnetic translations $U_a= U_a(B)$,  $ a \in \R^2$, defined  by
\begin{equation}
\left(U_a \psi\right)(x) = \e^{-i \frac B2  (x_2a_1 - x_1 a_2)} \psi(x -a),
\end{equation}
give a projective unitary representation of 
$\R^2$ on $\mathrm{L}^2(\mathbb{R}^2, {\mathrm{d}}x)$:
\begin{equation}
U_a U_b =   \e^{i \frac B2  (a_2b_1 - a_1 b_2)}  U_{a+b} =
\e^{i  B  (a_2b_1 - a_1 b_2)} U_b U_a, \quad a,b \in \Z^2. 
\end{equation}
We have  $U_a H_{B} U_a^* = H_B$ for all $a \in \R^2$, and  the following covariance relation
for magnetic translation by elements of $\Z^2$:
\begin{equation}\label{covariance}
U_a H_{B,\lambda,\omega} U_a^* = H_{B,\lambda,\tau_a\omega}
\quad \text{for all  $a \in \Z^2$}.
\end{equation}

It follows from ergodicity that that 
$H_{B,\lambda,\omega}$ has a nonrandom spectrum: there exists a nonrandom set $\Sigma_{B,\lambda}$ such that
$\sigma (H_{B,\lambda,\omega})=\Sigma_{B,\lambda} $ with probability one.
Moreover  the
decomposition of $\sigma (H_{B,\lambda,\omega})$ into pure point spectrum,
absolutely continuous spectrum, and singular continuous spectrum is also
independent of the choice of $\omega $ with probability one \cite{KM,CL,PF}.
In addition, the integrated density of states  $N(B,\lambda,E)$  is well defined 
and may be written as (cf.  \cite{HLMW1})
\begin{equation}\label{IDS11}
N(B,\lambda,E) = \E \left\{\tr \left\{\chi_0 P_{B,\lambda,E,\omega}  \chi_0\right\}\right\} .
\end{equation}
Here and throughout the paper, $\chi_{x}$ denotes the characteristic function of a cube of side length $1$ centered at $x \in \Z^{2}$.

The spectrum    of the  Landau
Hamiltonian $H_B$, denoted by $\Sigma_B$, consists of a sequence of infinitely
degenerate eigenvalues, the
Landau levels:
\begin{equation} \label{landaulevels}
\Sigma_B=\set{B_n :=(2n-1)B ,\quad n=1,2,\dotsc} .
\end{equation}
We also set  $B_0=-\infty$ for convenience.
Standard arguments (see Appendix~\ref{appLsp})  show  that 
\begin{equation} \label{splandau}
\Sigma_{B,\lambda} \subset
 \bigcup_{n=1}^\infty \mathcal{B}_n(B,\lambda) ,\quad\mbox{where}
\quad \mathcal{B}_n(B,\lambda)=
[B_n - \lambda M_1, B_n +\lambda M_2] . 
\end{equation}

For a given magnetic field $B >0$, disorder $\lambda \ge 0$ and energy $E \in \R$, 
the Fermi projection $P_{B,\lambda,E,\omega}$ is just the spectral
projection of the ergodic Landau Hamiltonian $H_{B,\lambda,\omega}$
onto energies $\le E$, i.e.,
\begin{equation}
P_{B,\lambda,E,\omega}= \chi_{(-\infty,E]}(H_{B,\lambda,\omega}).
\end{equation}
Estimates on the decay of  the operator kernel  of
the Fermi projection,
$$\left\{\chi_x P_{B,\lambda,E,\omega} \chi_y\right\}_{x,y \in\Z^2},$$  play an important role in the study of the Hall conductance.

To state these estimates we  introduce norms on random operators (see Subsection~\ref{aprandomop} for more details). 
A random operator $S_\omega$ is a strongly measurable 
map from the probability
space $(\Omega,\P)$ to bounded operators on 
$\mathrm{L}^2(\mathbb{R}^2, {\mathrm{d}}x)$.   
We set
\begin{align}\notag
\qnorm{S_\omega}_p  & := \left\{ \E \left\{ \tr |S_\omega|^p   \right\}\right\}^{\frac 1p}
= 
\left\lVert  \| S_\omega \|_p \right\rVert_{\text{L}^p(\Omega,\P) }\quad \text{for} \quad p \in [1,\infty),\\
\qnorm{S_\omega}_\infty  & :=
\left\lVert  \| S_\omega \| \right\rVert_{\text{L}^\infty(\Omega,\P) }.\label{opnorms}
\end{align}

The Hall conductance $\sigma_H(B,\lambda,E)$ is given by
\begin{equation}\label{sigmaH}
\sigma_H (B,\lambda,E) = - 2\pi i\,
\E \left\{\tr \left\{\chi_0 P_{B,\lambda,E,\omega}
\left[ \left[P_{B,\lambda,E,\omega},X_1\right], 
\left[P_{B,\lambda,E,\omega},X_2\right]\right]
\chi_0\right\}\right\},
\end{equation}
defined for $B >0$, $\lambda \ge 0$ and energy $E \in \R$ such that
\begin{equation} \label{welldef}
\qnorm{\chi_0 P_{B,\lambda,E,\omega}
\left[\left[ P_{B,\lambda,E,\omega},X_1\right],
\left[ P_{B,\lambda,E,\omega},X_2\right]\right] \chi_0}_1 < \infty.
\end{equation}
($X_i$ denotes the operator given by multiplication by the coordinate
$x_i$, $i=1,2$, and  $|X|$ the operator given by multiplication by $|x|$.)

A natural  condition for  \eqref{welldef} and quantization of the Hall conductance was given by Bellissard et al \cite{BES}:
\begin{equation}\label{K2}
\sum_{x  \in \Z^2} |x|^2\,
\qnorm{
\chi_x  P_{B,\lambda,E,\omega} \chi_0}_2^2
< \infty.
\end{equation} 
They showed the sufficiency of this  condition in an abstract $C^*$-algebra setting, from which they obtained existence and quantization of the Hall conductance for ergodic Landau Hamiltonians in a tight binding representation (ergodic Landau Hamiltonians).
This condition was also shown by
Bouclet and the authors \cite{BGKS} to be sufficient for a rigorous
derivation of \eqref{sigmaH} for ergodic Landau Hamiltonians as a Kubo formula.

Aizenman and Graf \cite{AG} gave a more elementary derivation of  the existence and quantization of the Hall conductance
for an ergodic Landau Hamiltonian $H_{B,\lambda,\omega}$ on $\ell^2(Z^2)$,  under the condition  \cite[condition (5.4)]{AG}, namely
\beq \label{AG54}
\sum_{x  \in \Z^2} |x| \set{\E \set{ \abs {\langle \delta_x, P_{B,\lambda,E,\omega}\delta_0 \rangle}^q}}^{\frac 1 q} < \infty \quad \text{for some} \quad q >2,
\eeq
which implies \eq{K2} in the discrete setting. 

In the discrete setting, given an interval where the integrated density of states is continuous, constancy of the Hall conductance  follows if either \eq{K2} or  \eq{AG54} holds with a  uniform bound in the interval \cite{BES,AG}.

On the continuum,  it is natural to work with  estimates on the the decay of $ \qnorm{
\chi_x  P_{B,\lambda,E,\omega} \chi_0}_2$.  In fact, it is known that for the Anderson-Landau Hamiltonian $ \qnorm{
\chi_x  P_{B,\lambda,E,\omega} \chi_0}_2$ exhibits sub-exponential in $x$ in the region of localization \cite[Theorem~3]{GKjsp},\cite[Eq.~(3.2)]{GKS}. We will prove that a sufficient condition for the  existence and quantization of the Hall conductance for ergodic Landau Hamiltonians is given by
\begin{equation}\label{newcond}
\sum_{x  \in \Z^2} |x|\,
\qnorm{
\chi_x  P_{B,\lambda,E,\omega} \chi_0}_2^\beta < \infty  \quad \text{for some} \quad \beta  \in (0,1).
\end{equation} 
We will also show that for an interval where the integrated density of states is continuous, we have constancy of the Hall conductance if \eq{newcond} holds with a locally bounded bound.
Note that \eq{newcond} implies \eq{K2}.

We consider the magnetic field-disorder-energy parameter space
\beq
\Xi=
\left\{(0,\infty)\times [0,\infty) \times \R\right\}\backslash 
\cup_{B \in (0,\infty)} \{(B,0)\times \Sigma_B \}  ;
\eeq
we exclude the Landau levels at no disorder. We give $\Xi$ the relative topology as a subset of $ \R^3$.  
Given a subset $\Phi \subset \Xi$,
we set
\begin{equation}\label{Phiup}
\Phi\up{B,\lambda} := \left\{E \in \R; \; (B,\lambda,E) \in \Phi\right\},
\end{equation}
with a similar definition for  $\Phi\up{B,E}$.

We now introduce  a (generalized)  ``localization length" $ L (B,\lambda,E)$,   based on \eq{newcond}. Given $\beta \in (0,1]$ and $(B,\lambda,E)\in \Xi$, we set
\beq \label{locL0}
L (B,\lambda,E) :=  \lim_{\beta \uparrow 1} L_\beta (B,\lambda,E),
\eeq 
where
 \beq \label{locLbeta0}
L_\beta (B,\lambda,E) :=\sum_{x  \in \Z^2} |x|\,
\qnorm{
\chi_x  P_{B,\lambda,E,\omega} \chi_0}_2^\beta \quad \text{for} \quad  \beta \in (0,1].
\eeq
We will also need `localization lengths" that take into account what happens near  $(B,\lambda,E)$.  We let
\begin{align}\label{locL0+}
L_+ (B,\lambda,E)& :=  \lim_{\beta \uparrow 1} L_{\beta +} (B,\lambda,E), \\
L _{+}\up{B,\lambda}(E)& := \lim_{\beta \uparrow 1} L _{\beta +}\up{B,\lambda}(E), \label{locL0+B}
\end{align}
where 
\begin{align}\label{locL+}
L _{\beta+}(B,\lambda,E)&:= \inf_{\substack{\Phi \ni (B,\lambda,E)\\
\Phi \subset \Xi\  \text{open} }}\
\sup_{(B^\prime,\lambda^\prime,E^\prime) \in \Phi} L_\beta (B^\prime,\lambda^\prime,E^\prime),\\ \label{locL+2}
L _{\beta +}\up{B,\lambda}(E)&:= \inf_{\substack{I \ni E\\
I \subset \R \  \text{open} }}\
\sup_{E^\prime \in I} 
L_\beta  (B,\lambda,E^\prime).
\end{align}
The justification of  the definitions \eq{locL0},  \eq{locL0+} and   \eq{locL0+B}, that is, the existence of the limits, is found in Subsection~\ref{subslocleng}. Note that $ L_1 (B,\lambda,E)<\infty$ implies \eq{K2}, and that  in general we only have   $ L_1 (B,\lambda,E) \le  L (B,\lambda,E)$.

We also define the subsets of $\Xi$ where these ``localization lengths" are finite:
\begin{equation}\begin{split}
\Xi_{\#}&:= \left\{ (B,\lambda,E) \in \Xi; \quad {\#}(B,\lambda,E) < \infty \right\},
\quad  \# = L, L_+,L_\beta, L _{\beta +} \,,\\
\Xi_{\#}\bp{B,\lambda} &: = \left\{ E \in \R; \quad {\#} \up{B,\lambda}(E) < \infty \right\},
\quad  \# =  L, L_+,L_\beta, L _{\beta +}\,.
\end{split} \end{equation}
$\Xi_{L_+}$ is,
by definition,  a relatively
open subset of $\Xi$, and $\Xi_{L_+}\bp{B,\lambda}$ is  an open subset of $\R$.  Note that   $\Xi_{\#}\bp{B,\lambda} \supset \Xi_{\#}\up{B,\lambda}$,  with $\Xi_{\#}\up{B,\lambda}$ defined as in \eqref{Phiup}, but we may not have equality. 

In Subsection~\ref{subslocleng} we show that the sets $\Xi_{\#} $ and $\Xi_{\#}\bp{B,\lambda}$,   $\# = L_\beta, L _{\beta +}$, are monotone increasing in $\beta \in (0,1]$, with
\beq\label{Xiunion}
\Xi_{L} = \bigcup_{\beta\in (0,1)} \Xi_{L_\beta},\quad  \Xi_{L_+} = \bigcup_{\beta\in (0,1)} \Xi_{L_\beta+}, \quad \Xi\bp{B,\lambda}_{L_+} = \bigcup_{\beta\in (0,1)} \Xi\bp{B,\lambda}_{L_\beta+}.
\eeq

Note that
\begin{equation}\label{NSq}
{\Xi}_{\text{NS}} :=
\left\{ (B,\lambda,E) \in {\Xi}; \; E \notin  \Sigma_{B,\lambda}\right\}\subset  \Xi_{L+}\, ;
\end{equation}
${\Xi}_{\text{NS}}$ being the  \emph{region of no spectrum}.

We are now ready to state our main results.

\begin{theorem}\label{sigmathm0}  Let  $H_{B,\lambda,\omega}$ be an ergodic  Landau Hamiltonian.  Then
 the Hall conductance $\sigma_H(B,\lambda,E)$ is defined and integer valued on  
$\Xi_{L}$.
In addition, $\sigma_H(B,\lambda,E)$ is
locally bounded on  $\Xi_{L_+}$ and on each $\Xi_{L_+}\bp{B,\lambda}$.
\end{theorem}

We  set $\sigma_H\up{B,\lambda}(E):=\sigma_H(B,\lambda,E)$,  $N\up{B,\lambda}(E):=N(B,\lambda,E)$, and $L_+\up{B,\lambda}(E):=L_+(B,\lambda,E)$.

\begin{theorem}\label{sigmathmN0}  Let  $H_{B,\lambda,\omega}$ be an ergodic  Landau Hamiltonian.  If for a given $(B,\lambda) \in (0,\infty)\times [0,\infty)$  the integrated density of states $N\up{B,\lambda}(E)$ is continuous in $E$,  then
 the Hall conductance $\sigma_H\up{B,\lambda}(E)$ is  continuous on  $\Xi_{L_+}\bp{B,\lambda}$.  In particular,
$\sigma_H\up{B,\lambda}(E)$ is constant on each connected component 
of $\Xi_{L_+}\bp{B,\lambda}$.
\end{theorem}

If  we have  
\begin{equation} \label{gapcond}
\lambda {(M_1 + M_2)}<  {2B},
\end{equation}
it follows from \eq{splandau} that the bands $\mathcal{B}_n(B,\lambda)$ are
disjoint,  and the spectral gaps remain open. We will refer to \eqref{gapcond} as the
\emph{disjoint bands condition}; it clearly holds
under  weak disorder and/or strong magnetic
field.

\begin{corollary} \label{maincorErg0} Let  $H_{B,\lambda,\omega}$ be an ergodic  Landau Hamiltonian.  Suppose  the integrated density of states $N\up{B,\lambda}(E)$ is continuous in $E$ for  all $(B,\lambda) \in (0,\infty)\times [0,\infty)$  satisfying  
the disjoint bands condition \eq{gapcond}. Then  for  all such  $(B,\lambda)$  the  ``localization length"
$L_+\up{B,\lambda}(E)$  diverges near each Landau level:
for
each $n=1,2,\ldots$ there exists an energy 
$E_n(B,\lambda) \in \mathcal{B}_n(B,\lambda) $ such that
\begin{equation}\label{divergent0}
L_{+}\bp{B,\lambda}(E_n(B,\lambda)) = \infty.
\end{equation}
\end{corollary}

For the Anderson-Landau Hamiltonian  $H\up{A}_{B,\lambda,\omega}$ we can say more. Following \cite{GKduke,GKjsp,GKS} we introduce the region of dynamical localization. (It  was called the strong insulator region in \cite{GKduke} and the  region of complete localization in \cite{GKjsp}.) This can be done in many equivalent ways, as shown in \cite{GKduke,GKjsp}, but for the purposes of this paper  we define it  by the decay of the Fermi projection, using \cite[Theorem~3 and following comments]{GKjsp}: The region of dynamical localization $\Xi_{\mathrm{DL}}$  consists of those $(B,\lambda,E) \in \Xi$ 
for which there exists an open interval $I \ni E$ 
such that
\begin{equation} \label{fermidecay}
\sup_{E^\prime \in I} \qnorm{\chi_{x}
P_{B,\lambda,E^\prime,\omega}
\chi_{0}}_2  
\leq  C_{I,B,\lambda} \pa{1 + \abs{x}}^{-\eta_1} \quad \text{for all} \quad x \in \Z^2,
\end{equation} 
where   $\eta_1>0$  is  a  fixed   number that can be calculated from the proof of \cite[Theorem~3]{GKjsp}.  (The condition stated  in \cite[Theorem~3]{GKjsp} is of the form
\begin{equation} \label{fermidecay2}
\mathbb{E}\left\{\sup_{E^\prime \in I} \left\|\chi_{x}
P_{B,\lambda,E^\prime,\omega}
\chi_{0}\right\|_2^2\right\}
\leq   C_{I,B,\lambda} \pa{1 + \abs{x}}^{-\eta_1} \quad \text{for all} \quad x \in \Z^2,
\end{equation} 
but an inspection of the proof shows that it can be replaced by \eq{fermidecay}.)
Its complement in $\Xi$ will be called   the region of dynamical delocalization: $\Xi_{\mathrm{DD}}:= \Xi \setminus \Xi_{\mathrm{DL}}$. (See \cite{GKS} for background, definitions, and discussion.)
It follows that  that there exists $\beta_1 \in (0,1)$ such that
\beq \label{XiDLXiLB}
\Xi_{\mathrm{DL}}\up{B,\lambda} =  \Xi_{L_{\beta_1 +}}\bp{B,\lambda} \subset     \Xi_{L_+}\bp{B,\lambda}. 
\eeq
Moreover, the integrated density of states $N(B,\lambda,E)$ of the  the Anderson-Landau Hamiltonian is jointly H\"older-continuous  in $(B,E)$ for $\lambda >0$ \cite{CHKR}. ( $N(B,\lambda,E)$ is actually Lipshitz continuous in $E$ for fixed $(B,\lambda)$ \cite{CHK2}.)
Thus \eq{divergent0} implies \cite[Eq.~(2.20)]{GKS}, that is,
\begin{equation}\label{notempty}
\Xi_{\text{DD}}\up{B,\lambda} \cap \mathcal{B}_n(B,\lambda)\not=
 \emptyset,
\end{equation}
 and hence Corollary~\ref{maincorErg0}
provides a new proof for  \cite[Theorems~2.1 and 2.2]{GKS}.

We actually have more. Using the characterization of  $\Xi_{\mathrm{DL}}$ as the region of applicability of the multiscale analysis \cite{GKduke}, we can  get the constant $C_{I,B,\lambda}$ in \eq{fermidecay} locally bounded in $B$ and $\lambda$, obtaining
\beq \label{XiDLXiL}
\Xi_{\mathrm{DL}} =  \Xi_{L_{\beta_1 +}} \subset     \Xi_{L_+}. 
\eeq

For the Anderson-Landau Hamiltonian we have a slightly stronger version of Theorems~\ref{sigmathm0}  and \ref{sigmathmN0}.

\begin{theorem}\label{sigmathm20}  Let  $H\up{A}_{B,\lambda,\omega}$ be the Anderson-Landau Hamiltonian.  Then
 the Hall conductance $\sigma_H(B,\lambda,E)$ is defined and integer valued on
$\Xi_{L}$,  
and H\"older-continuous on $\Xi_{L_+}$.
In particular,
$\sigma_H(B,\lambda,E)$ is constant on each connected component 
of $\Xi_{L_+}$.

It follows that  on $\Xi_{\mathrm{DL}} $, the region of dynamical localization , the Hall conductance $\sigma_H(B,\lambda,E)$ is defined, integer valued,
and  constant on each connected component .
\end{theorem}

The results in this article for the Anderson-Landau Hamiltonian go beyond   \cite[Theorems~2.1 and 2.2]{GKS};  they show the existence of a dynamical metal-insulator transition, in the sense of \cite{GKduke}, inside the Landau bands  of  the Anderson-Landau Hamiltonian
in cases when the disjoint bands condition does not hold and  the spectral gaps are closed.
We give a simple example in the next theorem.

As shown in  \cite{GKduke}, the region of dynamical localization  $\Xi^{(B,\lambda)}_{\mathrm{DL}}$ can be characterized as follows. To measure `dynamical localization' we introduce
\begin{equation}\label{moment}
M_{B,\lambda,\omega}(p,\mathcal{X},t)  = 
\left\|  {\langle} x {\rangle}^{\frac p 2}
{\mathrm{e}^{-i tH_{B,\lambda,\omega} }}
\mathcal{X}(H_{B,\lambda,\omega}) \chi_0 
\right\|_2^2 ,
\end{equation}
the
random moment of order
$p\ge 0$ at time $t$ for the time evolution  in the Hilbert-Schmidt norm,  
initially spatially localized in the square of side one around the origin
(with characteristic function $\chi_0$), and ``localized" 
in energy by the function $\mathcal{X}\in C^\infty_{c,+} (\mathbb{R})$. (Notation: $\langle x\rangle:= \sqrt{1 +\abs{x}^2}$.)
Its time averaged expectation is given by
\begin{equation}   \label{tam}
\mathcal{M}_{B,\lambda}( p ,\mathcal{X}, T )   =  
\frac1{T} \int_0^{\infty}
\mathbb{E}\left\{ M_{B,\lambda,\omega}(p,\mathcal{X},t)\right\}
{\mathrm{e}^{-\frac{t}{T}}} \,{\rm d}t .
\end{equation}
It is proven in  \cite{GKduke} that $ \Xi^{(B,\lambda)}_{\mathrm{DL}}$ is the set of energies $E$ for which there exists 
$\mathcal{X}\in C^\infty_{c,+} (\mathbb{R})$ with $\mathcal{X}\equiv 1$ on some open interval containing $E$, $\alpha \ge 0$, and $p > 4\alpha + 22$, such that
\beq \label{slowdec}
\liminf_{T \to \infty} \frac 1 {T^\alpha}\mathcal{M}_{B,\lambda}( p ,\mathcal{X}, T ) < \infty,
\eeq
in which case it is also shown in \cite{GKduke} that  \eq{slowdec} holds for any $p\ge 0$ with $\alpha=0$.

\begin{theorem}\label{closedgaps0} Let  $H\up{A}_{B,\lambda,\omega}$ be an Anderson-Landau Hamiltonian as in \eq{ALH}-\eq{potVL}, where  the   common probability
distribution $\mu$ has  density
\beq \label{rhoeta}
\rho(s)= \tfrac  {\eta +1}  2  \pa{1 - \abs{s}}^\eta  \chi_{[-1.1]}(s), \quad \eta >0, 
\eeq
and the single-site potential $u$   satisfies 
\beq\label{covcond0}
0< U_- \le U(x):= \sum_{i \in\mathbb{Z}^2}  u(x-i) \le  1,  \quad\text{with $U_-$  a constant}.
\eeq
Let $B>0$.  Then:\\
\noindent{(i)}  The spectral gaps are all closed for $\lambda \ge  \frac 1 {U_-} B$:
 \beq
\Sigma_{B,\lambda}= [ E_{0}(B,\lambda), \infty) \quad \text{for} \quad \lambda \ge  \tfrac 1 {U_-} B,
\eeq
where  $ E_{0}(B,\lambda):=\inf \Sigma_{B,\lambda} \in (B-\lambda,B- \lambda U_- )$.\\

\noindent{(ii)}   Let  $\widehat{\lambda} >  \frac 1 {U_-} B  $,  and $ \delta \in (0,B)$.  Set
\beq \begin{split}\label{JnB}
J_n(B)&:=\pa{{B}_n + \delta, {B}_{n+1} - \delta},  \quad n\in \N ,\\
J_0(B)&:=\pa{ -\infty , B -  \delta} \subset (-\infty, B).
\end{split}\eeq
Then for all $N \in \N$ there exists
${\eta}_N>0$ such that,  taking  $\eta \ge {\eta}_N$, for all $\lambda \in [0, \widehat{\lambda}]$  we have 
\beq\label{Jloc}
J_n(B)  \subset \Xi_{\mathrm{DL}}\up{B,\lambda}  \quad \text{for all} \quad  \lambda \in [0, \widehat{\lambda}], \quad n=0,1,2,\ldots,N  . 
\eeq
Moreover, 
for all 
$ \lambda \in [0, \widehat{\lambda}]$ there exists 
\begin{align}\label{EnB}
E_n(B,\lambda) \in
\pb{B_n -\delta, {B}_n +\delta }\cap \Xi_{\mathrm{DD}}\up{B,\lambda}\quad \text{for} \quad n=1,2,\ldots,N. 
\end{align}
In particular, for    $n=1,2,\ldots,N$ we have   $L_{+}\bp{B,\lambda}(E_n(B,\lambda)) = \infty$, and
for every  
$\mathcal{X}\in C^\infty_{c,+} (\mathbb{R})$  with
$\mathcal{X} \equiv 1$  on some open interval  $J\ni E_n(B,\lambda) $
and  $p>24$,  we have 
\begin{equation}\label{momentgrowth9}
\mathcal{M}_{B,\lambda}( p ,\mathcal{X}, T )   \ge \
C_{p,\mathcal{X}} \, T^{\frac p4 - 6}
\end{equation}
for all  $T \ge 0$  with  $  C_{p,\mathcal{X}} > 0 $.
\end{theorem}

Note that for all $\lambda \in [\frac 1 {U_-} B, \widehat{\lambda}]$ all the spectral gaps are closed, but we still  show existence  of  at least
one dynamical mobility edge near the first $N$ Landau levels,
namely a boundary  point between the regions of dynamical localization and dynamical delocalization.

Another application of the results in this paper can be found in a companion article \cite{GKM}, which considers an Anderson-Landau Hamiltonian  $H\up{A}_{B,\lambda,\omega}$ 
as in \eq{ALH}-\eq{potVL}, but with a   common probability
distribution $\mu$ which has    a bounded density $\rho$  with  $\supp \rho= \R$ and fast decay: 
\beq \label{rhodecay}
\rho(\omega) \le \rho_0 \exp(-|\omega|^\alpha) \quad \text{for some $\rho_0\in(0,+\infty)$ and $\alpha>0$}.
\eeq
(In particular, $\mu$ may have a Gaussian distribution.) The random potential $V_\omega$ is now an unbounded ergodic potential, but   $H_{B,\lambda,\omega}$ is  essentially self-adjoint on $\mathcal{C}_c^\infty(\R^d)$ with probability one,
and  we have  (see \cite{BCH})
\beq\label{BCHsp}
\Sigma_{B,\lambda}=\R  \quad \text{for all} \quad \lambda>0,
\eeq
where $\Sigma_{B.\lambda}$ is the  spectrum of $H_{B,\lambda,\omega}$ with probability one.

It is shown in \cite{GKM} that the main results of this paper,
and in particular Theorems~\ref{sigmathm0},  \ref{sigmathmN0} and \ref{sigmathm20}, as well as the relevant results from \cite{GKduke},
hold for these Anderson-Landau Hamiltonians with $\supp \mu=\R$ (and hence unbounded potentials). Note that \eq{XiDLXiL} is still valid, although its proof  must be modified, taking into account that the Wegner estimate  can be controled as $\lambda \to 0$ for intervals that do not contain Landau levels. The fact that the Landau gaps are immediately filled up as soon as the disorder is turned on implies that  the approach used in \cite{GKS} and in Corollary~\ref{maincorErg0}  is  not applicable. Proving the existence of a dynamical transition in that case requires the full set of conclusions of Theorem~\ref{sigmathm20}, namely that the Hall conductance is integer valued and continuous on connected components of $\Xi_{L_+}$, as used in the proof of Theorem~\ref{closedgaps0}.  The continuity of the Hall conductance for arbitrary small $\lambda$ (in order to let $\lambda$ go to zero) given by Theorem~\ref{sigmathm20} is required.
A result similar to Theorem~\ref{closedgaps0}(ii) is proved in \cite{GKM}: given $n\in \N$, there is   at least
one dynamical mobility edge near the first $N$ Landau levels for small $\lambda$.  It can be stated as follows.

\begin{theorem}[\cite{GKM}]\label{thmlimit2}   Let  $H_{B,\lambda,\omega}$ be a  random Landau Hamiltonian   as in \eq{ALH}-\eq{potVL}, but with a   common probability
distribution $\mu$ which has    a bounded density $\rho$  with  $\supp \rho= \R$
and  \eq{rhodecay}, so \eq{BCHsp} holds for all $\lambda>0$.
Let $B>0$.  Then,
for
each $n\in \N$, there exists $\lambda(n)>0$, such that for $\lambda \in (0,\lambda(n)]$  there exist 
$E_{n}\up{\pm} (B,\lambda)$, with $B_n -B < E_{n}\up{-}(B,\lambda)<B_n <E_{n}\up{+}(B,\lambda)<B_n +B$,
\begin{align} 
\left \lvert  {E}_{n}\up{\pm}(B,\lambda)  - B_n  \right \rvert  
\le  K_n(B)\lambda
\abs{\log \lambda}^{\frac 1 \alpha} \to  0 \quad 
\text{as $\lambda \to 0$},
\label{lambda20}
\end{align}
with a finite  constant  $ K_n(B) $, and 
  \beq
  \pa{E_{n}\up{+}(B,\lambda),(E_{n+1}\up{-}(B,\lambda)}\subset  \Xi_{\mathrm{DL}}\up{B,\lambda}.
\eeq
We also have   $\pa{-\infty, E_{1}\up{-}(B,\lambda)}\subset  \Xi_{\mathrm{DL}}\up{B,\lambda}$ for  $\lambda \in (0,\lambda(0)]$, $\lambda(0)>0$.

Moreover, for $\lambda \in (0,\min \set{\lambda(n-1),\lambda(n)})$ there exists
\begin{align}\label{EnB234}
E_n(B,\lambda) \in
\pb{E_{n}\up{-}(B,\lambda), E_{n}\up{+}(B,\lambda)}\cap \Xi_{\mathrm{DD}}\up{B,\lambda},
\end{align}
 and hence  \eq{momentgrowth9} holds
for every  
$\mathcal{X}\in C^\infty_{c,+} (\mathbb{R})$  with
$\mathcal{X} \equiv 1$  on some open interval  $J\ni E_n(B,\lambda) $
and  $p>24$.
\end{theorem}

We collect some technicalities in Section~\ref{secttech}.
 In Section~\ref{sectHall} we study the Hall conductance, proving Theorem~\ref{sigmathm0}.
Section~\ref{seccontHall} is devoted to the continuity of the Hall conductance:  Theorem~\ref{sigmathmN0} is proved in Subsection~\ref{subsecerg}, and  the stronger version for   Anderson-Landau Hamiltonians, Theorem~\ref{sigmathm20}, is proved in Subsection~\ref{secHallALH}.   Corollary~\ref{maincorErg0} is proven in Section~\ref{secdelergodic}. Dynamical delocalization (and a dynamical metal-insulator transition) for  Anderson-Landau Hamiltonians with closed spectral gaps is shown in  Section ~\ref{secApALH}, where we prove Theorem~\ref{closedgaps0}.  In Appendix~\ref{appLsp}  we prove a useful lemma about the spectrum of Landau Hamiltonians with bounded potentials.  The spectrum of the   Anderson-Landau Hamiltonian is discussed in  Appendix~\ref{apSp}.

\section{Technicalities}\label{secttech}

\subsection{Norms on random operators and Fermi projections}\label{aprandomop}
 Given $p \in [1,\infty)$,
$\T_p$ will denote  the Banach space of bounded operators $S$
on $\mathrm{L}^2(\mathbb{R}^2, {\mathrm{d}}x)$ with
$\| S \|_{\T_p}=\| S \|_p := \left(\tr |S|^p\right)^{\frac 1p} < \infty$. 
A random operator $S_\omega$ is a strongly measurable 
map from the probability
space $(\Omega,\P)$ to bounded operators on 
$\mathrm{L}^2(\mathbb{R}^2, {\mathrm{d}}x)$.   Given $p \in [1,\infty)$,
we set
\begin{equation}
\qnorm{S_\omega}_p:=
\left\{ \E \left\{  \| S_\omega \|_p^p   \right\}\right\}^{\frac 1p}= 
\left\lVert  \| S_\omega \|_{\T_p} \right\rVert_{\text{L}^p(\Omega,\P) },
\end{equation}
and
\begin{equation}
\qnorm{S_\omega}_\infty:=
\left\lVert  \| S_\omega \| \right\rVert_{\text{L}^\infty(\Omega,\P) }.
\end{equation}
These are norms on random operators, note that
\begin{equation}\label{qcomp}
\qnorm{S_\omega}_q \le \qnorm{S_\omega}_\infty^{\frac{q-p} q} \qnorm{S_\omega}_p^{\frac pq}
\quad \text{for $1\le p\le q <\infty$},
\end{equation}
and they satisfy Holder's inequality:
\begin{equation}\label{Holdersineq}
\qnorm{S_\omega T_\omega}_r\le  \qnorm{S_\omega}_p \qnorm{T_\omega}_q
\quad \text{for $r,p,q \in [1,\infty]$ with $\tfrac 1 r= \tfrac 1 p + \tfrac 1q$}.
\end{equation}
In particular, if $\qnorm{S_\omega}_\infty\le 1$, we have
\begin{equation}\label{qcomp5}
\qnorm{S_\omega}_q \le  \qnorm{S_\omega}_2^{\frac 2 q}
\quad \text{for $2\le p\le q <\infty$},
\end{equation}

\subsection{Operator kernels of Fermi projections}  Let $H_{B,\lambda,\omega}$ be an  ergodic Landau Hamiltonian 
for a given magnetic field $B >0$, disorder $\lambda \ge 0$ and energy $E \in \R$. 
We consider  the operator kernel  of
the Fermi projection $P_{B,\lambda,E,\omega}= \chi_{(-\infty,E]}(H_{B,\lambda,\omega})$,
$\left\{\chi_x P_{B,\lambda,E,\omega} \chi_y\right\}_{x,y \in\Z^2}$,
and set
\begin{equation}\begin{split}\label{trest0}
\kappa_p (B,\lambda, E) &\equiv 
\qnorm{\chi_0 P_{B,\lambda,E,\omega} \chi_0}_p \quad \text{for} \ p \in [1,\infty], \\
\kappa_{1,\infty}(B,\lambda, E) 
&\equiv \left\lVert \tr \left\{\chi_0 P_{B,\lambda,E,\omega} \chi_0\right\}
\right\rVert_{\text{L}^\infty (\Omega,\P)}.
\end{split}\end{equation}
Note that $\kappa_{1,\infty}(B,\lambda, E)$ is locally bounded on $\Xi$ 
(e.g., \cite{BGKS}), and hence also $\kappa_p (B,\lambda, E)$, since
$\kappa_\infty (B,\lambda, E) \le 1$ and for $p\in [1,\infty)$ we have
\begin{equation}
\kappa_p (B,\lambda, E)  \le
\qnorm{\chi_0 P_{B,\lambda,E,\omega} \chi_0}_1^{\frac 1p}\le
\left\{\kappa_{1,\infty}(B,\lambda, E)\right\}^ {\frac 1p} .
\end{equation}
In addition,  we have 
\begin{align}\label{trest9}
\qnorm{\chi_0  P_{B,\lambda,E,\omega}  }_{p} 
\begin{cases}
=\qnorm{\chi_0  P_{B,\lambda,E,\omega} \chi_0 }_{\frac p 2}^{\frac 12}
=\left\{\kappa_{\frac p2} (B,\lambda, E)\right\}^{\frac 12} 
&\text{if $p\in [2,\infty)$}\\[1.5mm]
= \qnorm{\abs{\chi_0 P_{B,\lambda,E,\omega}}^{\frac 12} }_{2p} ^2 
\le \kappa_{p}(B,\lambda, E) &\text{if $p\in [1,\infty)$}
\end{cases},
\end{align}
and thus, given $x\in \Z^2$,  for all $p\in [1,\infty)$ we have
\begin{align}\label{trest}
\qnorm{\chi_0  P_{B,\lambda,E,\omega} \chi_x }_p& \le
\qnorm{\chi_0  P_{B,\lambda,E,\omega}  }_{2p} 
\qnorm{  P_{B,\lambda,E,\omega} \chi_x }_{2p}=
\kappa_p (B,\lambda, E) .
\end{align}

It follows from \eq{IDS11} that  
\beq  \label{IDS12}
N(B,\lambda,E)= \kappa_1 (B,\lambda, E) .
\eeq
Note that  
\beq\label{N=0}
N (B,\lambda,E)=0 \quad \Longleftrightarrow  \quad  \qnorm{
\chi_x  P_{B,\lambda,E,\omega} \chi_0}_2=0 \quad \text{for all} \quad  x\in \Z^2.
\eeq

\subsection{"Localization lengths"}\label{subslocleng} We will now justify  the definitions \eq{locL0},  \eq{locL0+} and   \eq{locL0+B}.

To justify \eq{locL0}, we must show that the limit exists in $[0,\infty)$.  Given $\beta \in (0,1]$ and $(B,\lambda,E)\in \Xi$, let
 \beq \label{locLbeta}
\widetilde{L}_\beta (B,\lambda,E) :=N (B,\lambda,E)^{1 -\beta} L_\beta (B,\lambda,E),
\eeq
where  $N (B,\lambda,E)$ is as in \eq{IDS11}.   It follows from \eq{trest}
that $ \widetilde{L}_\beta (B,\lambda,E)$  is monotone decreasing in $\beta \in (0,1]$, so we can define 
\beq \label{locL}
\widetilde{L} (B,\lambda,E) := \inf_{\beta \in (0,1)} \widetilde{L}_\beta (B,\lambda,E)= \lim_{\beta \uparrow 1} \widetilde{L}_\beta (B,\lambda,E).
\eeq
It is an immediate consequence of \eq{locLbeta}   and  \eq{locL} (cf. \eq{N=0}) that $L (B,\lambda,E)$ is well defined and 
\beq
L (B,\lambda,E)= \widetilde{L} (B,\lambda,E).
\eeq

 The definitions \eq{locL0+} and   \eq{locL0+B} are justified in a similar way.    As before
\beq \begin{split}
\widetilde{L }_{\beta+}(B,\lambda,E)&:= N (B,\lambda,E)^{1 -\beta}L _{\beta+}(B,\lambda,E),\\ 
\widetilde{L} _{\beta +}\up{B,\lambda}(E)&:= N (B,\lambda,E)^{1 -\beta}L _{\beta +}\up{B,\lambda}(E),
\end{split}\eeq
are seen to be  monotone decreasing in $\beta \in (0,1]$, so we have
\begin{align}
L_+ (B,\lambda,E)& = \inf_{\beta \in (0,1)}\widetilde{L }_{\beta+}(B,\lambda,E)= \lim_{\beta \uparrow 1} \widetilde{L }_{\beta +} (B,\lambda,E), \\
L _{+}\up{B,\lambda}(E)& = \inf_{\beta \in (0,1)}\widetilde{L} _{\beta +}\up{B,\lambda}(E)= \lim_{\beta \uparrow 1} \widetilde{L} _{\beta +}\up{B,\lambda}(E).
\end{align}

It follows that that the sets $\Xi_{\#} $ and $\Xi_{\#}\bp{B,\lambda}$,   $\# = L_\beta, L _{\beta +}$, are monotone increasing in $\beta \in (0,1]$, and we have \eq{Xiunion}

\subsection{Auxiliary ``localization lengths"}
Although the ``localization lengths" \\  $L(B,\lambda,E)$ and $L_+(B,\lambda,E)$  give a convenient way to write our main theorems, in the proofs it will be more convenient to work with auxiliary  ``localization lengths"  based on the norms for  random operators introduced in \eq{opnorms} with  $p \in [2,\infty)$.  They can be thought  of an adaptation to the 
continuum 
(and to two parameters) of 
\cite[condition (5.4)]{AG}.  If  $q \in [1,\infty)$, $J \subset [1,\infty)$, we define the 
following ``localization lengths" for $(B,\lambda,E)\in \Xi$:
\begin{equation}\begin{split} \label{qcond0}
\ell_q(B,\lambda,E) &:= \sum_{x  \in \Z^2}\max\left\{  |x|,1\right\} \qnorm{
\chi_x  P_{B,\lambda,E,\omega} \chi_0}_q,\\
\ell_{q+}(B,\lambda,E)&:= \inf_{\substack{\Phi \ni (B,\lambda,E)\\
\Phi \subset \Xi\  \text{open} }}\
\sup_{(B^\prime,\lambda^\prime,E^\prime) \in \Phi} 
\ell_q(B^\prime,\lambda^\prime,E^\prime), \\
\ell_{q+}\up{B,\lambda}(E)&:= \inf_{\substack{I \ni E\\
I \subset \R \  \text{open} }}\
\sup_{E^\prime \in I} 
\ell_q(B,\lambda,E^\prime), \\
\ell_{J} (B,\lambda,E) &:=\inf_{q \in J} \ell_q(B,\lambda,E),\\
\ell_{J+} (B,\lambda,E) &:=\inf_{q \in J} \ell_{q+} (B,\lambda,E),\\
\ell_{J+} \up{B,\lambda}(E) &:=\inf_{q \in J} \ell_{q+}\up{B,\lambda}(E).
\end{split}\end{equation}

While  the quantity in
\cite[(5.4)]{AG} is monotone increasing in $q \in [1,\infty)$, 
the ``localization lengths"  $\ell_q(B,\lambda,E)$
cannot be compared for different $q$'s.  Another difference
is that \cite[condition (5.4)]{AG}  implies the equivalent of \eqref{K2} in the
lattice, but $\ell_q(B,\lambda,E) < \infty$  only
implies \eqref{K2} if $q=2$.  

We also
define the subsets of $\Xi$ where these ``localization lengths" are finite:
\begin{equation}\begin{split}
\Xi_{\#}&= \left\{ (B,\lambda,E) \in \Xi; \quad \ell_{\#}(B,\lambda,E) < \infty \right\},
\quad  \# = q, q+, J, J+\,,\\
\Xi_{\#}\bp{B,\lambda} &= \left\{ E \in \R; \quad \ell_{\#} \up{B,\lambda}(E) < \infty \right\},
\quad  \# =  q+,  J+\,.
\end{split} \end{equation}
Note that we may have  $\Xi_{\#}\bp{B,\lambda}\not= \Xi_{\#}\up{B,\lambda}$,  with $\Xi_{\#}\up{B,\lambda}$ defined as in \eqref{Phiup}. However, $\Xi_{\#}\bp{B,\lambda}\supset  \Xi_{\#}\up{B,\lambda}$  and 
\begin{equation}
\Xi_{J} = \bigcup_{q \in J} \Xi_q\, , \qquad 
\Xi_{J+} = \bigcup_{q \in J} \Xi_{q+}\, , \qquad \Xi_{J+}\bp{B,\lambda} = \bigcup_{q \in J} \Xi_{q+}\bp{B,\lambda}\, .
\end{equation}
$\Xi_{J+}$ is,
by definition,  a relatively
open subset of $\Xi$, and $\Xi_{J+}\bp{B,\lambda}$ is  an open subset of $\R$.

If $q \in [2,\infty)$,   it follows  immediately from \eq{qcomp5} and \eq{trest0} that for all $(B,\lambda,E)\in \Xi$ we have 
\begin{align}
\ell_q(B,\lambda,E) & \le  \kappa_q (B,\lambda, E) + L_{\frac 2 q}(B,\lambda,E),\\
\ell_{q+}(B,\lambda,E) & \le \kappa_q (B,\lambda, E)+  L_{\frac 2 q +}(B,\lambda,E),\\
\ell_{q+}\up{B,\lambda}(E) & \le \kappa_q (B,\lambda, E)+ L_{\frac 2 q +}\up{B,\lambda}(E) .
\end{align}
It follows that 
\beq
\Xi_L \subset   \bigcap_{r >2} \Xi_{(2,r]} \quad \text{and} \quad \Xi_{L_+} \subset   \bigcap_{r >2} \Xi_{(2,r]+} .
\eeq

For the Anderson-Landau Hamiltonian  $H\up{A}_{B,\lambda,\omega}$ the following holds for  all large $ q_0$  (recall \eq{fermidecay}-\eq{XiDLXiL}):
\beq \begin{split}
\Xi_{\mathrm{DL}} &=  \bigcap_{q \in [1,\infty)} \Xi_{q+}=  \Xi_{q_0+},\\
\Xi_{\mathrm{DL}}\up{B,\lambda} &=  \bigcap_{q \in [1,\infty)} \Xi_{q+}\bp{B,\lambda}=  \Xi_{q_0+}\bp{B,\lambda} .
\end{split}\eeq

\section{Existence and quantization of the Hall conductance}\label{sectHall}

Theorem~\ref{sigmathm0} is an immediate consequence of the following theorem.

\begin{theorem}\label{sigmathm}  Let  $H_{B,\lambda,\omega}$ be an ergodic  Landau Hamiltonian.  Then
 the Hall conductance $\sigma_H(B,\lambda,E)$ is defined on
$\Xi_{[2,\infty)}$ with the bound
\begin{equation}\label{uvsigmathm}
\abs{\sigma_H(B, \lambda,E)}  \le  4 \pi  \inf_{\substack{ q \in [2,\infty)\\ \frac 1 p +\frac 2 q=1}}\set{\kappa_{p} (B,\lambda, E)
\left\{\ell_q(B,\lambda,E)\right\}^2}
< \infty.
\end{equation}
It follows that $\sigma_H(B,\lambda,E)$ is
locally bounded on  $\Xi_{[2,\infty)+}$ and on each $\Xi_{[2,\infty)+}\bp{B,\lambda}$.
Moreover,    the Hall conductance $\sigma_H(B,\lambda,E)$ is
integer valued on $\Xi_{(2,3]}$.
\end{theorem}

Theorem~\ref{sigmathm} will proved by the following lemmas.

Given $x \in \R^2$, we set
$\hat{x}$ to be the discretization of $x$, i.e., the unique element of $\Z^2$ such that
$x_i \in [\hat{x}_i - \frac 12,\hat{x}_i + \frac 12)$, $1=1,2$.  We let
$\hat{X}_i$ denote the operator given by multiplication by
$\hat{x}_i$,  and note that $\hat{X}_i \chi_u = u_i \chi_u$
for each $u \in \Z^2$, i.e., $ \hat{X}_i = \sum_{x \in \Z^2} x \chi_x$,  and note
\begin{equation}\label{xhat}
\norm{X_i -\hat{X}_i } \le \frac 12, \quad
\norm{\abs{X} -\lvert{\hat{X}}\rvert} \le \frac {\sqrt{2}}2.
\end{equation}
If $(B,\lambda,E) \in \Xi$ and $q\in [1,\infty) $, it follows that
\begin{equation}
\qnorm{\lvert\hat{X}\rvert\,  P_{B,\lambda,E,\omega} \chi_0 }_q
\le \ell_q(B,\lambda,E), 
\end{equation}
and hence, using \eqref{xhat}, and \eqref{trest9} we get
\begin{align} \label{qcond2}
\qnorm{\abs{X}  P_{B,\lambda,E,\omega} \chi_0 }_q \le 
\ell_q(B,\lambda,E) + \kappa_{ q }(B,\lambda, E)\le 2\ell_q(B,\lambda,E).
\end{align}
It follows  that, with $i=1,2$,
\begin{align}\label{[PXhat]}
\qnorm{[  P_{B,\lambda,E,\omega} , \hat{X}_i ]\chi_0}_q & 
\le \ell_q(B,\lambda,E),\\
\qnorm{[  P_{B,\lambda,E,\omega} , {X}_i ]\chi_0}_q & 
\le 3 \ell_q(B,\lambda,E).\label{[PX]}
\end{align}

We conclude, using covariance, that for  $\P$-a.e. $\omega$, 
$\hat{X}_i  P_{B,\lambda,E,\omega} \chi_u$
and ${X}_i  P_{B,\lambda,E,\omega} \chi_u$, and hence also
$[  P_{B,\lambda,E,\omega} , \hat{X}_i ]\chi_u$ and 
$[  P_{B,\lambda,E,\omega} , {X}_i ]\chi_u$,
are bounded operators for all $(B,\lambda,E) \in \Xi_{[1,\infty)}$, $u \in \Z^2$,
$i=1,2$.

We now define a modified Hall conductance, with $\hat{X}_i$ substituted
for ${X}_i$:
\begin{equation}\label{hatsigmaH}
\hat{\sigma}_H (\lambda,E) = - 2\pi i\,
\E \left\{\tr \left\{\chi_0 P_{B,\lambda,E,\omega}
\left[ \left[P_{B,\lambda,E,\omega},\hat{X}_1\right], 
\left[P_{B,\lambda,E,\omega},\hat{X}_2\right]\right]
\chi_0\right\}\right\},
\end{equation}
defined for $(B,\lambda, E) \in\Xi$ such that
\begin{equation} \label{hatwelldef}
\qnorm{\chi_0 P_{B,\lambda,E,\omega}
\left[\left[ P_{B,\lambda,E,\omega},\hat{X}_1\right],
\left[ P_{B,\lambda,E,\omega},\hat{X}_2\right]\right] \chi_0}_1
< \infty.
\end{equation}

\begin{lemma} \label{lemmacond}
The Hall conductances $\sigma_H(B,\lambda,E)$ and  
$\hat{\sigma}_H(B,\lambda,E)$ are  defined
on the set $\Xi_{[2,\infty)}$.
Moreover, for all $(B,\lambda,E) \in\Xi_{[2,\infty)}$ we have
\begin{align}\label{uvsigma}
\sigma_H(B,\lambda,E)&= \hat{\sigma}_H(B,\lambda,E) \\
\notag & = -2\pi i \!\!\!
\sum_{u,v\in \Z^2} (u_1 v_2 - u_2 v_1)
\E\left \{\tr \left \{\chi_0 P_{B,\lambda,E,\omega} \chi_u P_{B,\lambda,E,\omega}
\chi_v P_{B,\lambda,E,\omega}\chi_0 \right\}\right\},
\end{align}
with
\begin{equation}\begin{split}\label{uvsigma2}
\abs{\sigma_H(B, \lambda,E)} &\le
4 \pi
\sum_{u,v\in \Z^2} \left\lvert u\right\rvert \left\lvert v \right\rvert
\qnorm{\chi_0 P_{\lambda,E,\omega} \chi_u P_{\lambda,E,\omega}
\chi_v P_{\lambda,E,\omega}\chi_0}_1\\
& \le  4 \pi\kappa_{p} (B,\lambda, E)
\left\{\ell_q(B,\lambda,E)\right\}^2
< \infty
\end{split}\end{equation}
for all $ q \in  [2,\infty)$ and  $ \frac 1 p +\frac 2 q=1$.
\end{lemma}

\begin{proof}  Let  $(B,\lambda,E) \in \Xi_q$ 
for some $q \in [1,\infty)$.  Writing $P_\omega$ for $P_{B,\lambda,E,\omega}$, we have
\begin{align}
&\qnorm{\chi_0 P_{\omega}
\left[\left[ P_\omega,X_1\right],
\left[ P_\omega,X_2\right]\right] \chi_0}_1
\le \\
& \ \sum_{u \in \Z^2} \left\{\qnorm{\chi_0 P_\omega
\left[ P_\omega,X_1\right]\chi_u
\left[ P_\omega,X_2\right] \chi_0}_1
+ \qnorm{\chi_0 P_\omega
\left[ P_\omega,X_2\right]\chi_u
\left[ P_\omega,X_1\right] \chi_0}_1\right\} < \infty,\notag
\end{align}
since  may use the  Holder's inequality \eqref{Holdersineq}
with 
$ \frac 1 p +\frac 2 q=1$ to get
\begin{align}
& \sum_{u \in \Z^2} \qnorm{\chi_0 P_\omega
\left[ P_\omega,X_i\right]\chi_u
\left[ P_\omega,X_j\right] \chi_0}_1 \label{manip}\\
& \quad \le
\qnorm{\chi_0  P_{\omega}}_p
\sum_{u \in \Z^2} 
\qnorm{\left[ P_{\omega},X_i\right] \chi_u}_q
(|u| +1)
\qnorm{
\chi_u P_{\omega} \chi_0}_q  \notag\\
& \quad  \le\qnorm{\chi_0  P_{\omega}}_p
\qnorm{\left[ P_\omega, X_i\right] \chi_0}_q \sum_{u \in \Z^2} 
(|u| +1)\,
\qnorm{
\chi_u P_{\omega} \chi_0}_q  \notag\\
& \quad   \le  4 \kappa_{p} (B,\lambda, E)
\left\{\ell_q(B,\lambda,E)\right\}^2 \notag
< \infty
\end{align}
for $i.j=1,2$, where we used  covariance,  
\eqref{trest9}, \eqref{[PX]}, and \eqref{qcond0}.
Thus  $\sigma_H(B,\lambda,E)$ is
 defined
on the set $\Xi_{q}$, and similarly for 
$\hat{\sigma}_H(B,\lambda,E)$.

We will now show that
$\sigma_H(B,\lambda,E)= \hat{\sigma}_H(B,\lambda,E)$. 
To see that, note that
\begin{align}\label{sigma-sigma}
&\sigma_H(B,\lambda,E) - \hat{\sigma}_H(B,\lambda,E) =\\
& \quad
- 2\pi i\,
\E \left\{\tr \left\{\chi_0 P_\omega
\left[ \left[P_\omega,X_1 -\hat{X}_1\right], 
\left[P_\omega,{X}_2\right]\right]
\chi_0\right\}\right\}\notag\\
&\qquad \quad \quad +  2\pi i\,
\E \left\{\tr \left\{\chi_0 P_\omega
\left[ \left[P_\omega,\hat{X}_1\right], 
\left[P_\omega,X_2 -\hat{X}_2\right]\right]
\chi_0\right\}\right\}\notag .
\end{align}
We have
\begin{align}
&\E \left\{\tr \left\{\chi_0 P_\omega
\left[ \left[P_\omega,X_1 -\hat{X}_1\right], 
\left[P_\omega,{X}_2\right]\right]
\chi_0\right\}\right\} \notag
\\
& \quad \quad = \E \left\{\tr \left\{\chi_0 P_\omega(X_1 -\hat{X}_1)
\left(1 -P_\omega\right)
\left[P_\omega,{X}_2\right]
\chi_0\right\}\right\}
\label{cent1} \\
& \qquad \qquad \quad+ \E \left\{\tr \left\{\chi_0 
\left[P_\omega,{X}_2\right] \left(1 -P_\omega\right)
(X_1 -\hat{X}_1)P_\omega\chi_0\right\}\right\}\notag \\
& \quad \quad = \E \left\{\tr \left\{\chi_0(X_1 -\hat{X}_1)
\left(1 -P_\omega\right)
\left[P_\omega,{X}_2\right] P_\omega
\chi_0\right\}\right\} 
\label{cent2}\\
& \qquad \qquad \quad+ \E \left\{\tr \left\{\chi_0 
(X_1 -\hat{X}_1)P_\omega
\left[P_\omega,{X}_2\right] \left(1 -P_\omega\right)
\chi_0\right\}\right\}\notag \\
& \quad \quad = \E \left\{\tr \left\{\chi_0 
(X_1 -\hat{X}_1) \left[P_\omega,{X}_2\right] 
\chi_0\right\}\right\}=0,\label{cent3}
\end{align}
where in \eqref{cent3} we used centrality of trace, justified since $X_2 \chi_0$
is a bounded operator,  to go from  \eqref{cent2} to \eqref{cent3}
we used
\begin{equation}
(1- P_\omega)[P_\omega,X_2]P_\omega
+ P_\omega[P_\omega,X_2](1- P_\omega) =
[P_\omega,X_2],
\end{equation}
and
the passage from \eqref{cent1} to \eqref{cent2} can be justified as follows:
\begin{align} 
& \notag \E \left\{\tr \left\{\chi_0 P_\omega(X_1 -\hat{X}_1)
\left(1 -P_\omega\right)
\left[P_\omega,{X}_2\right]
\chi_0\right\}\right\}\\
& \notag\quad = \sum_{u \in \Z^2} 
\E \left\{\tr \left\{\chi_0 P_\omega\chi_u(X_1 -\hat{X}_1)
\left(1 -P_\omega\right)
\left[P_\omega,{X}_2\right]
\chi_0\right\}\right\}\\
&\quad = \sum_{u \in \Z^2} 
\E \left\{\tr \left\{\chi_u(X_1 -\hat{X}_1)
\left(1 -P_\omega\right)
\left[P_\omega,{X}_2\right]
\chi_0 P_\omega\chi_u\right\}\right\}\\
&\notag \quad = \sum_{u \in \Z^2} 
\E \left\{\tr \left\{\chi_0(X_1 -\hat{X}_1)
\left(1 -P_\omega\right)
\left[P_\omega,{X}_2\right]
\chi_{-u} P_\omega\chi_0\right\}\right\}\\
&\notag \quad =
\E \left\{\tr \left\{\chi_0(X_1 -\hat{X}_1)
\left(1 -P_\omega\right)
\left[P_\omega,{X}_2\right]
P_\omega\chi_0\right\}\right\},
\end{align}
with a similar  calculation for the other term in  \eqref{cent2}, where we used
the centrality 
of the trace and covariance (the absolute summability of all series can be
verified as in \eqref{manip}).  The second term in the right-hand-side of 
\eqref{sigma-sigma} is also equal to $0$ by a similar calculation, so we conclude
that  $\sigma_H(B,\lambda,E)= \hat{\sigma}_H(B,\lambda,E)$.

Since, with  $ \frac 1 p +\frac 2 q=1$, we have
\begin{align}
|u| |v|\qnorm{\chi_0 P_\omega \chi_u P_\omega
\chi_v P_\omega\chi_0}_1 \le 
|u| \qnorm{
\chi_0 P_\omega \chi_u}_q \qnorm{
\chi_0  P_\omega}_p |v| \qnorm{
\chi_v P_\omega \chi_0}_q,
\end{align}
the estimate \eqref{uvsigma2} follows from \eqref{qcond0} and \eqref{trest9}. 
The expression \eqref{uvsigma}
then follows for  ${\sigma}_H(B,\lambda,E)=\hat{\sigma}_H(B,\lambda,E)$
from \eqref{hatsigmaH}.
\end{proof}

Next, we will show that the Hall conductance $\sigma_H(\lambda,E)$ 
takes integer values on $\Xi_{(2,3]}$, following the approach  of 
Avron, Seiler and Simon  \cite{ASS}, as modified by 
Aizenman and Graf \cite{AG}.   Avron, Seiler and Simon  proved the result
for random Landau Hamiltonians at energies  
outside the spectrum, i.e., on   $\Xi_{\text{NS}}$.   Their argument was adapted 
to the lattice by Aizenman and Graf, who proved that the 
Hall conductance for the lattice model  takes integer values in the region where 
\cite[condition (5.4)]{AG} holds, i.e., on the lattice equivalent of
 $\Xi_{(2,3]}$. (On the lattice this result had been  proved earlier under
the  lattice equivalent of  condition \eqref{K2} 
by Bellissard, van Elst and Schulz-Baldes \cite{BES}.) 
We complete the circle by adapting 
Aizenman and Graf's argument back to the continuum.

Let $\Z^{2*}= (\frac 12,\frac 12) + \Z^2 $ denote the dual lattice to $\Z^2$.
Given $a \in \Z^{2*}$ we define  the complex valued function 
$\gamma_a(x)$ on $\R^2$ by
\begin{equation}
\gamma_a(x) = \frac {\hat{x}_1- a_1 + i(\hat{x}_2- a_2)}
{\left\lvert {\hat{x}}- {a}\right \rvert},
\end{equation}
and let
$\Gamma_a$ denote the unitary operator given by multiplication by the
function $\gamma_a(x)$.  Note that  
$\left\lvert {\hat{x}}- {a}\right \rvert \ge \frac {\sqrt{2}}{2}$
for all $x \in \R^2$.  We have the following estimate:
\begin{equation}\label{gammaest}
\left\lvert \gamma_a(x) - \gamma_a(y)\right \rvert \le \min \left\{
\left\lvert \hat{x}-\hat{y}\right \rvert 
\max \left\{\frac1{\left\lvert {\hat{x} - a}\right \rvert}, 
\frac1{\left\lvert {\hat{y} - a}\right \rvert}\right\},2\right\}\le
\min \left\{  4 \frac {\left\lvert \hat{x}-\hat{y}\right \rvert}
{\left\lvert {\hat{x} - a}\right \rvert},2 \right\}.
\end{equation}
(The first inequality can be found in  \cite{ASS}.  The second inequality can be seen as
follows:
if $ |\hat{x} - \hat{y}|\le \frac 12  |\hat{x}-a |$
we have $   |\hat{x}-a| - |\hat{y}-a|\le  |\hat{x} - \hat{y}| \le  \frac 12  |\hat{x}-a |$,
and hence $ |\hat{x}-a| \le2 |\hat{y}-a|$;  if 
 $ |\hat{x} - \hat{y}|  >\frac 12  |\hat{x}-a |$ we have
$ \frac {\left\lvert \hat{x}-\hat{y}\right \rvert}
{\left\lvert {\hat{x} - a}\right \rvert} > \frac 12 $, and hence 
$ 4 \frac {\left\lvert \hat{x}-\hat{y}\right \rvert}
{\left\lvert {\hat{x} - a}\right \rvert} > 2$.)

Given two orthogonal projections $P$ and $Q$ in a Hilbert space, such that $P-Q$ is compact,  the index of $P$ and $Q$ is defined by (cf. \cite[Section~2]{ASS})
\beq
\text{Index} (P,Q):= \dim \Ker  (P-Q-1)-  \dim \Ker  (Q-P-1).
\eeq
The index is a well defined integer since $P-Q$  compact implies that $ \dim \Ker (P-Q \pm 1)$ are both finite.  Note that in the case $P$ and $Q$ have finite rank we have
\beq
\text{Index} (P,Q)= \dim \Ran P - \dim \Ran Q=\tr (P-Q).
\eeq

\begin{lemma}  The Hall conductance $\sigma_H(B,\lambda,E)$ 
takes integer values on  $\Xi_{(2,3]}$. 
\end{lemma}

\begin{proof}
Let  $(B,\lambda,E) \in \Xi_q$ 
for some $q \in (2,3]$, and write $P_\omega$ for $P_{B,\lambda,E,\omega}$. 
As in  \cite{ASS,AG}, we prove that
for all $a \in  \Z^{2*}$ we have
\begin{equation}\label{|Index|}
\E\left( \left \lVert 
P_{\omega} - \Gamma_a P_\omega \Gamma_a^* \right\rVert_3\right) < \infty, 
\end{equation}
and hence for $\P$-a.e. $\omega$ 
the index of the orthogonal projections  $P_{\omega}$ and
$\Gamma_a P_\omega \Gamma_a^*$ (see \cite[Section 2]{ASS}), 
$\text{Index} (P_{\omega},\Gamma_a P_\omega \Gamma_a^* )  $, is the  finite 
integer given by
\begin{equation}\label{Index}
\text{Index} (P_{\omega},\Gamma_a P_\omega \Gamma_a^* ) = \tr 
\left(P_{\omega} - \Gamma_a P_\omega \Gamma_a^* \right)^3.
\end{equation}
Note that   $\text{Index} (P_{\omega},\Gamma_a P_\omega \Gamma_a^* )$
is independent of $a \in  \Z^{2*}$ \cite[Proposition 3.8]{ASS}, and hence it follows from 
the covariance relation \eqref{covariance} and properties of the index
(use \cite[Proposition 2.4]{ASS}) that for all $b \in  \Z^{2}$ we have
\begin{equation}\begin{split}
&\text{Index} (P_{\tau_b\omega},\Gamma_a P_{\tau_b\omega} \Gamma_a^* )
=\text{Index}
(U_b P_{\omega}U_b^*,\Gamma_a U_b P_{\omega}U_b^* \Gamma_a^* )\\
& \qquad = 
\text{Index} (P_{\omega},\Gamma_{a+b} P_\omega \Gamma_{a+b}^* )
= \text{Index} (P_{\omega},\Gamma_{a} P_\omega \Gamma_a^* ).
\end{split}\end{equation}
Since $\text{Index} (P_{\omega},\Gamma_a P_\omega \Gamma_a^* )$
is a measurable function by \eqref{Index}, it follows from
ergodicity that  it must be constant almost surely 
(see \cite[Proposition 8.1]{ASS}). 
In particular, this constant must be an integer, and, since constants are integrable,
\begin{equation}
\E \left\{\text{Index} (P_{\omega},\Gamma_a P_\omega \Gamma_a^* ) \right\}
=\text{Index} (P_{\omega},\Gamma_a P_\omega \Gamma_a^* ) 
\quad \text{for  $\P$-a.e. $\omega$}.
\end{equation}
is an integer, and the lemma will follow if we show
\begin{equation}\label{sigmaindex}
\sigma_H(B,\lambda,E) = \E \left\{
\text{Index} (P_{\omega},\Gamma_a P_\omega \Gamma_a^* ) \right\}.
\end{equation}

Let $T_\omega = P_{\omega} - \Gamma_a P_\omega \Gamma_a^* $.
We have
\begin{equation}
\left \lVert T_\omega \right\rVert_q \le
\sum_{y \in \Z^2}  \left \lVert \sum_{x \in \Z^2} 
\chi_{x +y}T_\omega \chi_x \right\rVert_q ,
\end{equation}
where
\begin{equation}\begin{split}
&\left \lVert \sum_{x \in \Z^2}  \chi_{x +y}T_\omega \chi_x \right\rVert_q^q
= \tr \left\lvert \sum_{x \in \Z^2}
\chi_x T_\omega^*\chi_{x +y}T_\omega \chi_x \right\rvert^{\frac q2}\\
&\qquad=\sum_{x \in \Z^2} \tr \left\lvert 
\chi_x T_\omega^*\chi_{x +y}T_\omega \chi_x \right\rvert^{\frac q2}
= \sum_{x \in \Z^2} \left \lVert   \chi_{x +y}T_\omega \chi_x \right\rVert_q^q,
\end{split}\end{equation}
and hence
\begin{equation}\label{AG411}
\left \lVert T_\omega \right\rVert_q \le
\sum_{y \in \Z^2}
\left( \sum_{x \in \Z^2} \left \lVert   \chi_{x +y}T_\omega \chi_x 
\right\rVert_q^q \right)^{\frac1q},
\end{equation}
which is the extension of \cite[Lemma 1]{AG} to the continuum.  
Note that if the right hand side of \eqref{AG411} is finite, then 
\begin{equation}\label{AG4119}
T_\omega = 
\sum_{y \in \Z^2}
\left( \sum_{x \in \Z^2}   \chi_{x +y}T_\omega \chi_x \right) \quad
\text{in $\T_q$},
\end{equation}
where  $\T_q$ is the Banach space of
compact operators with the norm $\| \ \|_q$, in the sense that for each
${y \in \Z^2}$
the series $ \sum_{x \in \Z^2} 
\chi_{x +y}T_\omega \chi_x  $ converges in $\T_q$, to, say, $T^{(y)}$ (but the series is not 
necessarily absolutely summable),  the series $\sum_{y \in \Z^2} T^{(y)}$
converges absolutely in $\T_q$, and $T= \sum_{y \in \Z^2} T^{(y)}$.

It follows from  \eqref{gammaest} that
\begin{equation}
\qnorm{ \chi_{x +y}T_\omega \chi_x }_q \le   
4\tfrac {\left\lvert {y}\right \rvert}
{\left\lvert {{x} - a}\right \rvert}
\qnorm{ \chi_{y}P_\omega \chi_0 }_q,
\end{equation}
and hence
\begin{equation}\begin{split} \label{|Index|2}
&\E \left(\left \lVert T_\omega \right\rVert_q \right) 
\le \sum_{y \in \Z^2}
\left( \sum_{x \in \Z^2}\qnorm{\chi_{x +y}T_\omega \chi_x }_q^q
\right)^{\frac1q}\\ 
& \quad \le  4 
\left(\sum_{x \in \Z^2} 
\tfrac 1 {\left\lvert {{x} - a}\right \rvert^q} \right)^{\frac1q}
\ell_q(B,\lambda,E)  < \infty,
\end{split}\end{equation}
where we used  $q >2$.
Since we also have  $q \le 3$, and $\|S \|_r \le \|S\|_s$ for any
 $ 1 \le s \le r < \infty$,
we note that \eqref{|Index|} follows from \eqref{|Index|2}.

It remains to prove \eqref{sigmaindex}. To do so, note that it follows from
\eqref{AG4119} and \eqref{Index} that
\begin{equation}\label{traceT3}
\text{Index} (P_{\omega},\Gamma_a P_\omega \Gamma_a^* ) =
\tr\,  T_\omega^3 = \sum_{u,v \in \Z^2}
 \left\{\sum_{x \in \Z^2}\tr \left(\chi_x T_\omega \chi_{x+u}
T_\omega \chi_{x+v} T_\omega \chi_{x}\right)\right\} 
\end{equation}
where the series in $x$ is at first only known to be convergent for each $u,v$, 
but not absolutely convergent, to, say, $\zeta(u,v)$,
and $ \sum_{u,v \in \Z^2}| \zeta(u,v) | < \infty   $.

To show that the series is actually absolutely convergent, we let $r$ be given
by  $\frac 1 r +\frac2 q =1$, so in particular $ q <r $, and note that, using \eqref{gammaest}, we have
\begin{align}
\sum_{u,v,x \in \Z^2}
&\E \left\{\tr \left\lvert\chi_x T_\omega \chi_{x+u}
T_\omega \chi_{x+v} T_\omega \chi_{x}\right\rvert\right\}\\
&\notag  \le
\sum_{u,v,x \in \Z^2} \qnorm{\chi_0 P_\omega \chi_{u}
P_\omega \chi_{v} P_\omega \chi_{0}}_1\tfrac {4|u|}{|x-a|}
\left\{2^{1 -\frac q r}\left(\tfrac {4|u-v|}{|x+u-a|}\right)^{\frac q r}\right\} 
\tfrac {4|v|}{|x-a|}\\
&\notag \le
64\sum_{u,v \in \Z^2} |u| |u-v|^{\frac q r} |v| \qnorm{\chi_0 P_\omega \chi_{u}
P_\omega \chi_{v} P_\omega \chi_{0}}_1
\sum_{a \in \Z^{2*}}\tfrac {1}{{|a|^2} {|u-a|^{\frac q r}} }< \infty,
\end{align}
since
\begin{equation}
\sum_{a \in \Z^{2*}}\tfrac {1}{{|a|^2} {|u-a|^{\frac q r}} }
\le \left(\sum_{a \in \Z^{2*}}\tfrac {1}{|a|^{\frac {6r}{3r-q}}}\right)^
{\frac{3r -q}{3r} }
\left(\sum_{a \in \Z^{2*}}\tfrac {1}{|a|^3}\right)^{\frac q {3r}} < \infty,
\end{equation}
and
\begin{align}\notag
&\sum_{u,v \in \Z^2}\!\! |u| |u-v|^{\frac q r} |v| \qnorm{\chi_0 P_\omega \chi_{u}
P_\omega \chi_{v} P_\omega \chi_{0}}_1\le 
\left\{\!\sup_{x \in \Z^2} |x|^{\frac q r} \qnorm{ \chi_{x}P_\omega \chi_0 
}_r \!\right\}\!\left\{\ell_q(B,\lambda,E)\right\}^2 \\
&  \quad \le 
\left\{\sup_{x \in \Z^2}  |x|\qnorm{ \chi_{x}P_\omega \chi_0 
}_q\right\}^{\frac q r}
\left\{\ell_q(B,\lambda,E)\right\}^2 
\le \left\{\ell_q(B,\lambda,E)\right\}^{2+ \frac q r}< \infty.
\end{align}

We can thus take expectations in \eqref{traceT3} obtaining
\begin{align}\label{traceT34} 
&\E \left\{
\text{Index} (P_{\omega},\Gamma_a P_\omega \Gamma_a^* ) \right\}
=
\sum_{u,v \in \Z^2} \E \left\{\tr \left(\chi_0 P_\omega \chi_{u}
P_\omega \chi_{v} P_\omega \chi_{0}\right)\right\} \times \\
&\qquad \times
\sum_{x \in\Z^2} (1- \gamma_a(x)\overline{\gamma_a}(x+u) ) 
(1- \gamma_a(x+u)\overline{\gamma_a}(x+v)) (1- \gamma_a(x+v)\overline{\gamma_a}(x)).\notag
\end{align}
On the other hand,
\begin{align}\label{connes}
&\sum_{x \in\Z^2}(1- \gamma_a(x)\overline{\gamma_a}(x+u) ) 
(1- \gamma_a(x+u)\overline{\gamma_a}(x+v)) (1- \gamma_a(x+v)\overline{\gamma_a}(x))\\
&\notag \quad= \sum_{a \in\Z^{2*}} (1- \gamma_a(0)\overline{\gamma_a}(u) ) 
(1- \gamma_a(u)\overline{\gamma_a}(v)) (1- \gamma_a(v)\overline{\gamma_a}(0))=- 2\pi i (u_1v_2 -u_2v_1)
\end{align}
by Connes formula as in \cite[Appendix F]{AG} -- see also 
\cite[Eqs. (4.14) and (5.1)]{AG}.

Thus \eqref{sigmaindex} follows from  \eqref{traceT34}, \eqref{connes},
and \eqref{uvsigma}.
\end{proof}

This completes the proof of Theorem~\ref{sigmathm}.

\section{Continuity of the Hall conductance}\label{seccontHall}

\subsection{Ergodic Landau Hamiltonians}\label{subsecerg}

Theorem~\ref{sigmathmN0}  follows immediately from the following   theorem.

\begin{theorem}\label{sigmathmN}  Let  $H_{B,\lambda,\omega}$ be an ergodic  Landau Hamiltonian.  If for a given $(B,\lambda) \in (0,\infty)\times [0,\infty)$  the integrated density of states $N\up{B,\lambda}(E)$ is continuous in $E$,  then
 the Hall conductance $\sigma_H\up{B,\lambda}(E)$ is  continuous on  $\Xi_{(2,\infty)+}\bp{B,\lambda}$.  In particular,
$\sigma_H\up{B,\lambda}(E)$ is constant on each connected component 
of $\Xi_{(2,3]+}\bp{B,\lambda}$.
\end{theorem}

To prove Theorem~\ref{sigmathmN} we will use the following lemma. 

\begin{lemma}\label{lemcont}
Let  $(B,E,\lambda) \in  \Xi_{q+}$ with   $q \in{(2,\infty)}$; set  $\frac 1 p + \frac 2 q =1$. Then there exists a neighborhood
$\Phi$   of $(B,E,\lambda)$ in $\Xi$, such that $\Phi \subset  \Xi_{q+}$,  and for all   $(B^\prime,\lambda^\prime,E^\prime) \in\Phi$  we have, with  $\sigma_H, \sigma_H^\prime,P_\omega, P_\omega^\prime$ for  
$\sigma_H(B,\lambda,E), \sigma_H(B^\prime,\lambda^\prime,E^\prime)$,
$P_{B,\lambda,E,\omega},P_{B^\prime,\lambda^\prime,E^\prime,\omega}$,
respectively.
\begin{equation}\label{contestimate}
\abs{ 
\sigma_H^\prime- \sigma_H} \le  C_{B,\lambda,E,q}  \left\{\sup_{u \in \Z^2}\qnorm{
\chi_0  \left(P_{\omega}^\prime -P_{\omega}\right) \chi_u}_1^{\frac 1p}\right\} \left\{\ell_{q+}(B,\lambda,E) \right\}^2.
\end{equation}
\end{lemma}

\begin{proof}
Given   $(B,E,\lambda) \in  \Xi_{q+}$ with   $q \in{(2,\infty)}$, there exists a neighborhood
$\Phi$   of $(B,E,\lambda) $ in $\Xi$
such that
\begin{equation} \label{qcond7}
\ell_q(B^\prime,\lambda^\prime,E^\prime)\le 
2\ell_{q+}(B,\lambda,E)<\infty
\end{equation} 
for any $(B^\prime,\lambda^\prime,E^\prime) \in \Phi$. 
(It follows that $\Phi \subset  \Xi_{q,+}$.)
We write $\sigma_H, \sigma_H^\prime,P_\omega, P_\omega^\prime$ for  
$\sigma_H(B,\lambda,E), \sigma_H(B^\prime,\lambda^\prime,E^\prime)$,
$P_{B,\lambda,E,\omega},P_{B^\prime,\lambda^\prime,E^\prime,\omega}$,
respectively. Using Lemma~\ref{lemmacond} and \eqref{hatsigmaH}, we have
\begin{align}\label{sig-sig}
\frac i {2\pi} \left(\sigma_H^\prime - \sigma_H\right)& = 
\E \left\{\tr \left\{\chi_0 \left(P_{\omega}^\prime -P_{\omega}\right)
\left[ \left[P_{\omega}^\prime,\hat{X}_1\right], 
\left[P_{\omega}^\prime,\hat{X}_2\right]\right]
\chi_0\right\}\right\}\\\notag
&\qquad +\E \left\{\tr \left\{\chi_0 P_{\omega}
\left[ \left[\left(P_{\omega}^\prime -P_{\omega}\right),\hat{X}_1\right], 
\left[P_{\omega}^\prime,\hat{X}_2\right]\right]
\chi_0\right\}\right\} \\  \notag
&\qquad + \E \left\{\tr \left\{\chi_0 P_{\omega}
\left[ \left[P_{\omega},\hat{X}_1\right], 
\left[\left(P_{\omega}^\prime -P_{\omega}\right),\hat{X}_2\right]\right]
\chi_0\right\}\right\}\\  \notag
& \equiv \sigma_1 + \sigma_2 + \sigma_3, 
\end{align}
where $\sigma_1, \sigma_2 , \sigma_3$ can be shown to be  well defined as in the proof of
Lemma~\ref{lemmacond}, and can be written similarly
to \eqref{uvsigma}.  Thus, with $\frac 1p + \frac 2 q=1$, where $p < \infty$ 
since $q >2$, we have
\begin{align} \notag
&\lvert \sigma_1 \rvert \le
\sum_{u,v\in \Z^2} \lvert (u_1- v_1) v_2 - (u_2-v_2) v_1\rvert
\E\left \{\tr \left \lvert\chi_0 \left(P_{\omega}^\prime -P_{\omega}\right)
\chi_u P_{\omega}^\prime
\chi_v P_{\omega}^\prime \chi_0 \right\rvert\right\}\\
& \quad \le 8 \left\{\sup_{u \in \Z^2}\qnorm{
\chi_0  \left(P_{\omega}^\prime -P_{\omega}\right) \chi_u
}_p\right\} \left\{\ell_{q+}(B,\lambda,E) \right\}^2   \label{sigma1}\\   \notag
& \quad \le 16 \left\{\sup_{u \in \Z^2}\qnorm{
\chi_0  \left(P_{\omega}^\prime -P_{\omega}\right) \chi_u}_1^{\frac 1p}\right\} \left\{\ell_{q+}(B,\lambda,E) \right\}^2,
\end{align}
with similar estimates for $|\sigma_2|$ and $|\sigma_3|$. 
The desired estimate \eq{contestimate}  now follows from
\eqref{sig-sig} and  \eqref{sigma1}.
\end{proof}

\begin{proof}[Proof of Theorem~\ref{sigmathmN}]

In view of Theorem~\ref{sigmathm}, it suffices to show that if for  a  given $(B,\lambda) \in (0,\infty)\times [0,\infty)$  the integrated density of states $N\up{B,\lambda}(E)$ is continuous in $E$,  then
 the Hall conductance $\sigma_H\up{B,\lambda}(E)$ is  continuous on  $\Xi_{(2,\infty)+}\bp{B,\lambda}$.  This follows immediately  from Lemma~\ref{lemcont}, since for $E_1 \le E_2$ we have, for all  $u \in \Z^2$,  
 \beq \begin{split}
 \qnorm{\chi_0  \left(P_{B,\lambda,E_2,\omega} -P_{B,\lambda,E_1,\omega} \right) \chi_u}_1
& \le \qnorm{\chi_0  \left(P_{B,\lambda,E_2,\omega} -P_{B,\lambda,E_1,\omega} \right) \chi_0}_1\\
&= N\up{B,\lambda}(E_2) -N\up{B,\lambda}(E_1).
\end{split}  \eeq
\end{proof}

\subsection{The Anderson-Landau Hamiltonian}\label{secHallALH} 

Theorem~\ref{sigmathm20} follows from the following theorem.

\begin{theorem}\label{sigmathm2}  Let  $H\up{A}_{B,\lambda,\omega}$ be the Anderson-Landau Hamiltonian.  Then
 the Hall conductance $\sigma_H(B,\lambda,E)$ is defined on
$\Xi_{[2,\infty)}$,  
integer valued on $\Xi_{(2,3]}$, and H\"older-continuous on $\Xi_{(2,\infty)+}$.
In particular,
$\sigma_H(B,\lambda,E)$ is constant on each connected component 
of $\Xi_{(2,3]+}$.
\end{theorem}

In view of Theorems~\ref{sigmathm}  and \ref{sigmathmN}, all that remains  to finish the proof of Theorem~\ref{sigmathm2} is to show that for the  Anderson-Landau Hamiltonian
the Hall conductance  $\sigma_H(B,\lambda,E)$ is  H\"older-continuous on $\Xi_{(2,\infty)+}$. 
This will  follow from Lemma~\ref{lemcont} and the following lemma, which improves
 on a result of  Combes, Hislop, Klopp, and Raikov \cite{CHKR}:   the integrated density  of states  of the   Anderson-Landau Hamiltonian
$N(B,\lambda,E)$ is jointly H\"older continuous in $(B,E)$ for $\lambda > 0$.
More precisely, they proved that  given given $\lambda > 0$, 
$\alpha,\delta  \in (0,1)$, and a compact set $Y \subset (0,\infty] \times \R$,
there exists a constant $C_{Y,\alpha,\delta}(\lambda)$ such that
\begin{equation}\label{NEB}
\abs{N(B^\prime,\lambda,E^\prime) - N(B,\lambda,E)}
\le C_{Y,\alpha,\delta}(\lambda)
\left( |B^\prime -B|^{\frac \alpha 4} + |E ^\prime - E|^{\delta } \right) 
\end{equation}
for all $(B,E), (B^\prime,E^\prime) \in Y$,  and the constant 
$ C_{Y,\alpha,\delta}(\lambda)$  is  locally bounded for$\lambda >0$.
(Although the fact that  
$ C_{Y,\alpha,\delta}(\lambda)$ is locally bounded is not explicitly
stated in \cite{CHKR}, it is implicit in the proof.)]
H\"older continuity in the energy was previously known 
in special cases \cite{CH,Wa,HLMW,CHK}. We
strengthen this result, proving  joint H\"older-continuity 
of $\chi_0 P_{B,\lambda,E,\omega}  \chi_0$ {in the
$\qnorm{ \ }_1$ norm with respect to $(B,E,\lambda)$.

\begin{lemma} \label{pjointHolder}   Let  $H\up{A}_{B,\lambda,\omega}$ be the Anderson-Landau Hamiltonian.
Fix  $\alpha,\delta, \eta  \in (0,1)$. Then, given a compact
subset 
$ K $ of  $ \Xi$,
there exists a constant
$C_{K,\alpha,\delta, \eta}$ such that
\begin{equation}\begin{split}
\label{jointHolder}
&\sup_{u \in \Z^2}\qnorm{ \chi_0\left(
P_{B^\prime,\lambda ^\prime,E ^\prime,\omega} - 
P_{B^{\prime\prime},\lambda^{\prime\prime},E ^{\prime\prime},\omega} \right) \chi_u}_1\\
& \qquad   \qquad  \qquad 
\le 
C_{K,\alpha,\delta,\eta}\left( |B^\prime -B|^{\frac \alpha 5}+
|E ^\prime - E^{\prime\prime }|^{\delta } 
+  |\lambda^\prime - \lambda^{\prime\prime }|^{\frac \eta 3}\right) 
\end{split}\end{equation} 
for all  
$(B^\prime,\lambda ^\prime,E ^\prime), (B^{\prime\prime},\lambda^{\prime\prime},E ^{\prime\prime}) 
\in K$.
\end{lemma}

Lemma~\ref{pjointHolder}   will follow from the above stated result of \cite{CHKR}
and  Lemma~\ref{oldsublemma} below.  Note that if $E^{\prime\prime} \le E^\prime$ we have
$ P_{B,\lambda,E ^\prime,\omega} - P_{B,\lambda,E ^{\prime\prime},\omega}\ge 0 $,
so the hypothesis of  Lemma~\ref{oldsublemma}  follow from \eqref{NEB}.

\begin{lemma}\label{oldsublemma}  Let  $H\up{A}_{B,\lambda,\omega}$ be the Anderson-Landau Hamiltonian. Let $\delta \in (0,1)$. 
Suppose that for every bounded interval 
$I$ and $(B,\lambda) \in (0,\infty)^2$ there exists a constant
$C_{I}(B,\lambda)$, locally bounded in $(B,\lambda)$, such that for
all  $E ^\prime, E^{\prime\prime }\in I$
we have
\begin{equation}\label{Ncont}
\qnorm{\chi_0 \left(
P_{B,\lambda,E ^\prime,\omega} - P_{B,\lambda,E ^{\prime\prime},\omega} \right)
\chi_0}_1
\le 
C_{I}(B,\lambda) |E ^\prime - E^{\prime\prime }|^\delta . 
\end{equation}
Given $ K= [B_1, B_2] \times [\lambda_1,\lambda_2]\times [E_1,E_2] \subset
\Xi   $,
there is a constant $C_{K}$, such that for all $ E \in [E_1, E_2]$ and $u \in \Z^2$
we have
\begin{align}\label{Hcdis}
\qnorm{\chi_0\left(
P_{B,\lambda ^\prime,E,\omega} - 
P_{B,\lambda^{\prime\prime},E,\omega} \right) \chi_u}_1
\le
C_{K}   |\lambda^\prime - \lambda^{\prime\prime} |^{\frac \delta {\delta +2}},
\end{align}
for all $B \in [B_1,B_2]$ and 
$ \lambda ^\prime,  \lambda^{\prime\prime}  \in   [\lambda_1,\lambda_2] $, and
\begin{align}\label{HcB}
\qnorm{\chi_0\left(
P_{B^\prime,\lambda ,E,\omega} - 
P_{B^{\prime\prime},\lambda,E,\omega} \right) \chi_u}_1
\le
C_{K}   |B^\prime-B^{\prime\prime} |^{\frac \delta {\delta +4}},
\end{align}
for all $B^\prime, B^{\prime\prime} \in [B_1,B_2]$ and 
$ \lambda   \in   [\lambda_1,\lambda_2] $.
\end{lemma}

\begin{proof} 
It suffices to consider the case  when
$ B_2 - B_1 <1$ and $ \lambda_2 -\lambda_1 <1$,
We set  $I= \left[E_1 - 1 , E_2\right] $.
Note that
\eqref{Ncont} holds for $(B,\lambda) \in [B_1, B_2] \times [\lambda_1,\lambda_2]$
and $E ^\prime, E^{\prime\prime }\in I$ with 
$C_I \equiv \sup_{(B,\lambda) \in [B_1, B_2] \times [\lambda_1,\lambda_2]}
C_{I}(B,\lambda)  < \infty$.
(This includes the case $\lambda_1=0$ with a slightly modified interval $I$, 
although this case is not included in
the hypothesis \eqref{Ncont}.  The reason is that since $K \subset \Xi$,
if $\lambda_1=0$ the interval 
$[E_1, E_2]$ cannot  contain any Landau level for $B \in [B_1,B_2]$. In this case we set
$I= \left[E_1 - \rho , E_2\right]$, where $0<\rho\le1$ is chosen  so $I$ also does not 
contain a Landau level for some  $B \in [B_1,B_2]$.  The proof applies also in this case
except that we take 
$ B_2 - B_1 <\rho $ and $ \lambda_2 -\lambda_1 <\rho$.)

We  fix a function 
$ f \in C^\infty(\R)$, such that $0 \le f(t) \le 1$,  $f(t)=1$ if $ t \le 0$,
and $f(t) =0$ if $t\ge 1$.

We prove \eqref{Hcdis} first.
Let $ E \in [E_1, E_2] $, $B \in [B_1,B_2]$, and 
$ \lambda ^\prime,  \lambda^{\prime\prime}  \in   [\lambda_1,\lambda_2] $.
We let   $\gamma=  |\lambda^\prime - \lambda^{\prime\prime} |^{\alpha}$,
where $\alpha \in (0,1)$ will be chosen later.
We set $g(t) = f\left( \frac { t- (E-\gamma)} {\gamma}\right)$; note
$ g \in C^\infty(\R)$, with $0 \le g(t) \le 1$, 
$g(t)=1$ if $ t \le
 E - \gamma$,
$g(t) =0$ if $t\ge E$. We write
\begin{align}\label{Pg}
& P_{B,\lambda ^\prime,E ,\omega}-  P_{B,\lambda ^{\prime\prime},E,\omega}= 
\left\{ P_{B,\lambda ^\prime,E ,\omega}-g(H_{B,\lambda ^\prime,\omega})\right\}\\
& \qquad  \qquad  \qquad 
+ \left\{g(H_{B,\lambda ^\prime,\omega}) -
g(H_{B,\lambda  ^{\prime\prime},\omega})\right\} +
\left\{g(H_{B,\lambda  ^{\prime\prime},\omega}) - 
P_{B,\lambda ^{\prime\prime},E,\omega}\right\}\notag .
\end{align}
By construction, for any $\lambda \ge 0$ we have
\begin{equation}
0 \le P_{B,\lambda,E ,\omega} - 
g(H_{B,\lambda ,\omega}) \le 
P_{B,\lambda ,E ,\omega} -  
P_{B,\lambda , E - \gamma,\omega}\, ,
\end{equation} 
and thus, for  $\lambda^\#= \lambda^\prime, \lambda^{\prime\prime}  $
and any  $u \in \Z^2$,
we have
\begin{align} \label{bounduseW}
&\qnorm{ \chi_0\left(
 P_{B,\lambda ^{\#},E ,\omega} - 
g(H_{B,\lambda ^\#,\omega}) \right) \chi_u}_1\\
& \notag \ \le \qnorm{ \chi_0\left(
P_{B,\lambda ^{\#},E ,\omega} - 
g(H_{B,\lambda ^\#,\omega}) \right)^{\frac 12} }_2
\qnorm{ \left(
 P_{B,\lambda ^{\#},E ,\omega} - 
g(H_{B,\lambda ^\#,\omega})  \right)^{\frac 12}  \chi_u
}_2\\
&\ \notag  =\qnorm{\chi_0\left(
 P_{B,\lambda ^{\#},E ,\omega} - 
g(H_{B,\lambda ^\#,\omega}) \right) \chi_0 }_1\\
& \ \notag \le \qnorm{\chi_0\left(
 P_{B,\lambda ^{\#},E ,\omega} - 
P_{B,\lambda ^{\#}, E - \gamma,\omega}\right) \chi_0 }_1 
\le  C_I \gamma^\delta.
\end{align}

We now estimate the middle term in the right hand side of \eqref{Pg}. 
Let $ R_{B,\lambda,B\omega}(z) = \left( H_{B,\lambda,\omega} -z\right)^{-1}$
be the resolvent.  Recall (e.g., \cite{BGKS}) that
\begin{equation}\label{Rtrace}
\left\lVert \chi_v R_{\lambda,B,\omega}(z) \right\rVert_2 \le
c_\lambda \frac { 1 + |z|}{\text{Im} z},
\end{equation}
with a constant $c_\lambda$ independent of $B$,  $v \in \Z^2$, and $\omega$, 
and locally bounded  in $\lambda$.
The
Helffer-Sj\"ostrand formula with a quasi analytic extension of $g$ of order $3$
(e.g., \cite{D}), combined with the resolvent equation and  \eqref{Rtrace},
yields
\begin{equation}\label{boundg}
\qnorm{\chi_0\left( g(H_{B,\lambda ^\prime,\omega}) -
g(H_{B,\lambda  ^{\prime\prime},\omega})\right) \chi_u }_1
\le
C \frac{|\lambda^\prime-\lambda^{\prime\prime}|}{\gamma^2},
\end{equation} 
where the constant $C$ depends only on $E_1,E_2,\lambda_1,\lambda_2$, our choice of the function $f$, and 
fixed parameters.

Thus, combining \eqref{Pg}, \eqref{bounduseW}, and \eqref{boundg}.
we get
\begin{align}\label{Cdeltaplus}
& \qnorm{\chi_0\left(
P_{\lambda ^\prime,E ^\prime,\omega} - 
P_{\lambda^{\prime\prime},E ^{\prime\prime},\omega} \right) \chi_u}_1\le
2   C_I \gamma^\delta +
 C \frac{|\lambda^\prime-\lambda^{\prime\prime}|}{\gamma^2}\\
& \quad = 2   C_I  |\lambda^\prime - 
\lambda^{\prime\prime} |^{\alpha \delta} +
C  |\lambda^\prime - \lambda^{\prime\prime} |^{1 - 2 \alpha} = (2C_I
+C)  |\lambda^\prime - \lambda^{\prime\prime} |^{\frac \delta {\delta +2}},
\notag
\end{align}
where we chose  $\alpha= \frac 1 {\delta +2}$ to
 optimize the bound.

To prove \eqref{HcB}, we start by repeating the above proof 
varying $B$ instead of $\lambda$.  The only difference is in the equivalent
of the estimate \eqref{boundg}.  Here we use \cite[Proposition 5.1]{CHKR},
observing that its proof   (note \cite[Eqs. (5.12) and (5.13)]{CHKR}) actually
proves the stronger result
\begin{equation}\label{boundgB}
\qnorm{\chi_0\left( g(H_{B^\prime,\lambda ,\omega}) -
g(H_{B ^{\prime\prime},\lambda ,\omega})\right) \chi_u }_1
\le
\tilde{C} \frac{|B^\prime-B^{\prime\prime}|}{\gamma^4},
\end{equation} 
where now  $\gamma = |B^\prime-B^{\prime\prime}|^\alpha$, and
the constant $\tilde{C}$ depends only on
$E_1,E_2,\lambda_1,\lambda_2, B_1,B_2$, our choice of the function $f$, and 
fixed parameters.  Proceeding as before, we see that in this case we should 
choose  $\alpha= \frac 1 {\delta +4}$, in which case we get \eqref{HcB}.
\end{proof}

\section{Delocalization for ergodic Landau Hamiltonians with open gaps}\label{secdelergodic}

We now prove  Corollary~\ref{maincorErg0} by proving the following theorem

\begin{theorem} \label{maincorErg} Let  $H_{B,\lambda,\omega}$ be an ergodic  Landau Hamiltonian.  Suppose  the integrated density of states $N\up{B,\lambda}(E)$ is continuous in $E$ for  all $(B,\lambda) \in (0,\infty)\times [0,\infty)$  satisfying  
the disjoint bands condition \eq{gapcond}. Then  for  all such  $(B,\lambda)$  the  ``localization length"
$\ell_{(2,3]+}\bp{B,\lambda}$  diverges near each Landau level:
for
each $n=1,2,\ldots$ there exists an energy 
$E_n(B,\lambda) \in \mathcal{B}_n(B,\lambda) $ such that
\begin{equation}\label{divergent}
\ell_{(2,3]+}\bp{B,\lambda}(E_n(B,\lambda)) = \infty.
\end{equation}
\end{theorem}

We start th eproof of Theorem~\ref{maincorErg}  by 
setting, for $n=1,2,\ldots$,
\begin{align}
\mathbb{G}_n & =\left\{ (B,\lambda,E) \in \Xi;\;  \lambda {(M_1 + M_2)}<  {2B},\, 
E \in (B_{n-1} +\lambda M_2, B_n -\lambda M_1 )\right\}.
 \end{align}
In view of
\eqref{splandau} and \eq{NSq}, we have
\begin{equation}
\bigcup_{n=1}^\infty \mathbb{G}_n = 
\Xi \setminus \bigcup_{B \in (0,\infty) }
\bigcup_{\lambda \in [0,\infty)}\bigcup_{n=1}^\infty 
\{(B,\lambda)\}\times\mathcal{B}_n(B,\lambda) \subset  {\Xi}_{\text{NS}}
\subset \Xi_{(2,3]+}\,.
\end{equation}

It is well known that 
$\sigma_H(B,0,E)=n$ if $E\in]B_{n},B_{n+1}[$ for all $n=0,1,2\dots$ 
\cite{ASS,BES}.
Given $n\in \N$ and $(B,\lambda_1,E)\in  \mathbb{G}_n$, we can find $\lambda_E > \lambda_1$ such that
$E\in \mathbb{G}_n\up{B,\lambda}$ for all $\lambda \in I= [0,\lambda_E[$.
It follows that, with probability one, 
\begin{equation}\label{contourint}
P_\lambda = - \tfrac 1 {2\pi i} \int_\Gamma R_\lambda(z) \,\di z  \quad 
\text{for all} \ \lambda \in I,
\end{equation}
where $P_\lambda = P_{B,\lambda,E,\omega}  $,
$R_\lambda(z) = (H_{B,\lambda,\omega}-z)^{-1}$, and $\Gamma$ is a
bounded contour such that  $\dist (\Gamma,\sigma (H_{B,\lambda,\omega})) \ge \eta >0$
for all $\lambda \in I$.
(Note $H_{B,\lambda,\omega} \ge B- \lambda_E M_1$ for all 
$\lambda \in I$.) 
It follows that there is a constant $K$ such that (cf. \cite[Proposition 2.1]{BGKS})
\begin{align}
\label{HSest}
\| R_\lambda(z) \chi_x\|_2 \le K \ \quad \text{for all} \ x \in \Z^2,
z \in \Gamma, \lambda \in I .
\end{align}
Given  $\lambda,\xi\in I$, it follows from \eqref{contourint} and the resolvent 
identity that
\begin{align}\label{taylor}
Q_{\lambda,\xi}:= P_{\xi} -P_\lambda =
\tfrac {(\xi -\lambda)} {2\pi i} \int_\Gamma R_\lambda(z) V R_\xi (z)\,\di z,
\end{align}
with $V= V_\omega$ (recall $\|V\| \le  \widetilde{M}:=\max \{M_1,M_2\}$).
Letting $\sigma_\lambda=\sigma_H(B,\lambda,E)$, it follows from Lemma~\ref{lemcont}
that for all $\lambda \in I$, 
 taking $\xi \in I$ in a suitable neighborhood of $\lambda$, we have 
\begin{equation}\label{contestimate3}
\abs{\sigma_\lambda- \sigma_\xi} \le  C_{B,\lambda,E}^\prime  \left\{\sup_{u \in \Z^2}\qnorm{
\chi_0  Q_{\lambda,\xi} \chi_u}_1^{\frac 13}\right\}\le C_{B,\lambda,E}^\prime \set{ \tfrac {\abs{\xi -\lambda}} {2\pi } \widetilde{M}\abs{\Gamma}K^2 }^{\frac 13},
\end{equation}
so $\sigma_\lambda$ is a continuous function of $\lambda$ in the interval $I$. By Theorem~\ref{sigmathm},  $\sigma_\lambda$ is constant in $I$, and hence we conclude that
\beq\label{sigma=n}
\sigma_H(B,\lambda,E)=\sigma_H(B,0,E)=n  \quad \text{for all} \quad (B,\lambda,E)\in  \mathbb{G}_n.
\eeq

Now, let $(B,\lambda)$  satisfy \eqref{gapcond}, and suppose
$\mathcal{B}_n(B,\lambda)\subset \Xi_{(2,3]+}\bp{B,\lambda} $
for some $n\in \N$.   We then have
\beq
(B_{n-1} +\lambda M_1,B_{n+1} - \lambda M_2)  
=\mathbb{G}_{n-1}\up{B,\lambda} \cup \mathcal{B}_n(B,\lambda) 
\cup  \mathbb{G}_{n}\up{B,\lambda}\subset  \Xi_{(2,3]+}\bp{B,\lambda}\, .
\eeq
Since  the integrated density of states $N\up{B,\lambda}(E)$ is assumed to be continuous in $E$, it follows from Theorem~\ref{sigmathmN} that the Hall
conductance 
$\sigma_{H} (B,\lambda,E)$ is constant on the interval $(B_{n-1} +\lambda M_1,B_{n+1} - \lambda M_2) $,and hence has the same value on the spectral gaps
$\mathbb{G}_{n-1}\up{B,\lambda}$ and $\mathbb{G}_{n}\up{B,\lambda}$,
which contradicts \eq{sigma=n}.  Thus we conclude that $\mathcal{B}_n(B,\lambda)$ cannot be a subset of $\Xi_{(2,3]+}\bp{B,\lambda}$, which proves Theorem~\ref{maincorErg}.

\section{Dynamical delocalization for the Anderson-Landau Hamiltonian with closed gaps}\label{secApALH}

In this section we prove Theorem~\ref{closedgaps0}.

Let  $H\up{A}_{B,\lambda,\omega}$ be an Anderson-Landau Hamiltonian as in \eq{ALH}-\eq{potVL}, with a   common probability
distribution $\mu$ with 
$\supp \mu=[-M_1, M_2]$ with $M_1,M_2 \in (0,\infty)$.
As shown in Appendix~\ref{apSp}, we have
\beq \label{splandau66777}
\Sigma_{B,\lambda} = \bigcup_{n\in \N} I_n(B,\lambda), \quad \text{where} \quad  I_n(B,\lambda)= [E_-(n,B,\lambda), E_+(n,B,\lambda)],
\eeq
where, for all $B>0$ and   $n \in\N$,   $\pm E_\pm(n,B,\lambda)$ are increasing, continuous functions of $\lambda >0$, depending on $u$ and   $M_1,M_2$, but not on other details of the measure $\mu$.      We   set $ E_+(0,B,\lambda)=-\infty$. We have \begin{gather} 
B_{n-  1} \le E_-(n,B,\lambda)  < B_n<   E_+(n,B,\lambda)\le B_{n+ 1} \quad \text{for all} \quad n\in \N,\\
B- \lambda M_1 \le  E_-(1,B,\lambda)= E_{0}(B,\lambda):=\inf \Sigma_{B,\lambda} < B ,\notag
\end{gather}
(Note that $ B- \lambda M_1 \le E_{0}(B,\lambda)$ follows from \eq{splandau}.)  In

   If \eqref{gapcond} holds, then   $E_+(n,B,\lambda) < E_-(n+1,B,\lambda)$ for all $n \in \N$  and the spectral gaps do not close.  If for some $n \in \N$ we have  $E_+(n,B,\lambda) \ge  E_-(n+1,B,\lambda)$, the $n$-th spectral gap $(B_n,B_{n+1})$ has closed, i.e.,
$[B_n,B_{n+1}]\subset \Sigma_{B,\lambda} $.

Let us now assume that the single-site potential  $u$ in \eq{potVL}  satisfies 
\beq\label{covcond}
0< U_- \le U(x):= \sum_{i \in\mathbb{Z}^2}  u(x-i) \le  1,
\eeq
for some constant $U_-$. (The upper bound is simply a normalization we had already assumed.) Then, as shown in Appendix~\ref{apSp}, we have
\begin{align}\label{locsp1}
B_n + \lambda M_2U_ -  & \le E_+(n,B,\lambda)  \quad \text{for} \quad \lambda \in \pa{0, \tfrac {2B}{M_2U_ -}},\\
B_n - \lambda  M_1 U_-   & \ge E_-(n,B,\lambda) \quad \text{for} \quad \lambda \in \pa{0, \tfrac {2B}{M_1 U_- }},\\
B - \lambda  M_1 U_-   & \ge E_-(1,B,\lambda)= E_{0}(B,\lambda) \quad \text{for all} \quad \lambda \ge 0.
\end{align}
It follows that if
\beq \label{gapclose1}
\lambda (M_1 + M_2) U_- \ge 2B,
\eeq
all the internal spectral gaps close, i.e.,
\beq
\Sigma_{B,\lambda} = [ E_{0}(B,\lambda), \infty).
\eeq
Theorem~\ref{closedgaps0}(i) is proven.

To prove  Theorem~\ref{closedgaps0}(ii), we assume \eq{rhoeta} 
and fix   $\widehat{\lambda} >  \frac 1 {U_-} B  $,  and $ \delta \in (0,B)$.  Let $J_n(B)$ be as in \eq{JnB}, we 
set
\beq \begin{split}\label{JnB1}
\widehat{J}_n(B)&:=\pa{{B}_n + \tfrac \delta 2, {B}_{n+1} - \tfrac \delta 2},  \quad n\in \N ,\\
\widehat{J}_0(B)&:=\pa{ -\infty , B -  \tfrac \delta 2} \subset (-\infty, B).
\end{split}\eeq

We will prove \eq{Jloc}  by a multiscale analysis.  The multiscale analysis is carried on for the finite volume operators defined in \cite[Section~4 and 5]{GKS}; the Anderson-Landau Hamiltonian satisfies all the requirements for the multiscale analysis plus  a Wegner estimate  \cite[Sections~4 and 5]{GKS}.  We take scales $L \in L_B \N$, where  $L_B\ge 1$ is defined in \cite[Eq.~(5.1)]{GKS}, and consider boxes $\Lambda_L(x)= x + [-\frac L 2, \frac L 2)^2$, $x \in \R^2$, and let   $\widetilde{\Lambda}_L(x)= \Lambda_L(x)\cap \Z^2$. We define finite volume operators $H_{B,\lambda,0,L,\omega}$ on  $ \mathrm{L}^2(\Lambda_L(0))$ as in \cite[Eq.~(5.2)]{GKS}:
\begin{equation}\begin{split}  \label{landauh2} 
H_{B,\lambda,0,L,\omega}& = H_{B,0,L} +
\lambda  V_{0,L,\omega} \quad \mathrm{on} \quad
\mathrm{L}^2(\Lambda_L(0)),\\
V_{0,L,\omega}(x)&=  \sum_{i\in  \widetilde{\Lambda}_{L-\delta_u}(0)}
\omega_i \,u(x-i),
\end{split}\end{equation}
where $ H_{B,0,L}$ is defined in  \cite[Sections 5]{GKS} and  $\supp u \subset  \pa{-\frac {\delta_u} 2, \frac {\delta_u} 2}^2$, and then define  $H_{B,\lambda,\omega,x,L}$ for all $ x \in \Z^2$ 
by   \cite[Eq.~(4.3)]{GKS}.  (We prescribed periodic boundary condition for the
(free) Landau Hamiltonian  at the square centered at
$0$, and used the magnetic translations to define the finite volume
operators in all other squares by  \cite[Eq.~(4.3)]{GKS}; in the square
centered at  $ x \in \Z^2$ the potential  $V_{x,L,\omega}$ is exactly as 
in \eqref{landauh2} except that the sum is now over 
$i\in  \widetilde{\Lambda}_{L-\delta_u}(x) $.) 

A   Wegner estimate is given in    \cite[Theorem~5.1]{GKS} and extended in \cite[Theorem~4.3]{CHK2}; 
note that the constants in the Wegner estimate  can be chosen uniformly in $\lambda \in [\lambda_1.\lambda_2]$  if $\lambda_1>0$.   It follows that for a closed interval  
$I \subset (B_n,B_{n+1})$, $n=0,1,2,\ldots$, they can be chosen uniformly in $\lambda \in 
[0,\widehat{\lambda} ]$. (But note that the constants will depend on the interval $I$, and hence for $I=\widehat{J}_n(B)$ they will depend on $n$.)
But one has to be careful in the multiscale analysis, since $\norm{\rho}_\infty$ appears in the Wegner estimate,   \eq{rhoeta}  gives $\norm{\rho}_\infty= \frac {\eta +1} 2$, and we will prove \eq{Jloc} for $\eta$ sufficiently large.

All these issues can be taken in consideration by applying the  finite volume criterion for localization given in
\cite[Theorem~2.4]{GKgafa},  in a similar  way  to the application in 
 \cite[Proof of Theorem~3.1]{GKgafa}. 

We write  $\Lambda=\Lambda_L(x)$, $H_{B,\lambda,L,\omega}=H_{B,\lambda,x,L,\omega}$, etc. 
If  $ \lambda \abs{\omega_i}  \le \frac \delta 2$ for all $i \in \widetilde{\Lambda}$, then we have by Lemma~\ref{lemmaLsp} (it also applies to finite volume operators) that 
\begin{equation} \label{splandau5599}
\sigma\pa{H_{B,\lambda,L,\omega}} \subset
 \bigcup_{n=1}^\infty \pb{B_n -\tfrac \delta 2, {B}_n +\tfrac \delta 2 }.
  \end{equation}
We have    
 \beq\begin{split}
\inf_{\lambda \in [0, \widehat{\lambda}]}  \P\set{\lambda \abs{ \omega_i} \le \tfrac \delta 2 \quad \text{for all  $i \in \widetilde{\Lambda}$}}
&\ge 1 - L^2 \,  \P\set{\widehat{\lambda} \abs{\omega_0} >\tfrac \delta 2}\\
& = 1- L^2 \pa{1- \tfrac \delta{2 \widehat{\lambda}}}^\eta 
\end{split}\eeq
where  $\frac \delta{2 \widehat{\lambda}} < \frac {U_- }{2}  \le\frac 1 2$

 Given $\omega$ satisfying \eq{splandau5599},  $E \in J_n(B)$ implies  $\dist\pa{E, \sigma\pa{H_{B,\lambda,L,\omega}}}> \frac \delta 2$. Let $R_{B,\lambda,L,\omega}(E) =\pa{H_{B,\lambda,L,\omega}-E}^{-1}$.
  It follows from the Combes estimate (cf. \cite[Theorem~1]{GKdecay}; note that the estimate holds for finite volume operators with periodic boundary condition with uniform constants for large enough volumes  using the distance on the torus, cf. 
   \cite[Lemma 18]{FKac} and  \cite[Theorem 3.6]{KK1}) that 
   \beq
   \|\chi_x R_{B,\lambda,L,\omega}(E) \chi_y \| \le \tfrac {C_1}{\delta} \e^{-C_2 \delta L }\quad \text{for all $x,y \in \widetilde{\Lambda}$ with $\abs{x-y} \ge \tfrac L {10}$},
   \eeq 
 where $C_1,C_2 >0$ are constants, depending only  on $n$, $B$, $u$.

Let us fix $n\in \N$ and prove that 
$ J_n(B)  \subset \Xi_{\mathrm{DL}}\up{B,\lambda}$  for all  $\lambda \in [0, \widehat{\lambda}]$. (The case $n=0$ can be handled in a similar manner.)   We take the constants in the Wegner estimate  valid for subintervals of $\widehat{J}_n(B)$, uniformly in $\lambda \in [0,\widehat{\lambda} ]$.    Thus, if we have \eq{splandau5599}, we will have the condition whose probability is estimated in 
\cite[Eq.~(2.17)]{GKS}  if
\beq\label{217}
L^9 \tfrac {C_1}{\delta} \e^{-C_2 \delta L } <  \frac { C_3}{\eta +1},
\eeq
where $C_3$ is another constant  depending only  on $n$, $B$, $u$, and $\delta$.

We now take $L_0(n)$ satisfying  \cite[Eq.~(2.16)]{GKS} and large enough for the Wegner estimate, and for $L_0 \ge L_0(n)$ we set 
\beq
\eta(n,L_0)= 1 + \tfrac{C_3 \delta} {2 C_1} L_0^{-9}  \e^{C_2 \delta L_0 },
\eeq
so \eq{217}  holds with $L=L_0$ and $\eta=\eta(n,L_0)$ . Since
\beq
\lim_{L_0 \to \infty} L_0^2 \pa{1- \tfrac \delta{2 \widehat{\lambda}}}^{ \eta(n,L_0)}=0
\eeq
Thus we can find $\eta(n) >0$ such that for all $\eta \ge \eta(n)$ there exists $L_0(\eta)\ge L_0(n)$
for which we have \cite[Eq.~(2.17)]{GKS}, so $E \in J_n(B)$  implies  $E \in \Xi_{\mathrm{DL}}\up{B,\lambda}$.

Thus given $N \in \N$, letting 
$\eta_N =\max_{n=0,1,2,\ldots,N} {\eta}(n)$, we have \eq{Jloc} for $\eta \ge \eta_N$.

Since  the Hall conductance 
$\sigma_H(B,0,E)=n$ if $E\in(B_{n},B_{n+1})$ for all $n=0,1,2\dots$ 
\cite{ASS,BES},
it follows from  Theorem~\ref{sigmathm20} that for $\eta \ge \eta_N$ we have 
\beq
\sigma_H(B,\lambda,E)= n \quad \text{for all} \quad (\lambda,E) \in 
[0, \widehat{\lambda}]\times J_n(B).
\eeq
We now proceed as in \cite[Proof of Theorem~2.2]{GKS}, using again  Theorem~\ref{sigmathm20} (here we could also use Theorem~\ref{sigmathmN0}), to conclude that
for   $n=1,2,\ldots,N$ we have $E_n(B,\lambda) \in
\pb{B_n -\delta, {B}_n +\delta }$ with  $L_{+}\bp{B,\lambda}(E_n(B,\lambda)) = \infty$, so we have \eq{EnB}, and \eq{momentgrowth9} follows from \cite[Theorem~2.11]{GKduke}, as in  \cite[Theorem~2.2]{GKS}.

Theorem~\ref{closedgaps0} is proven.

\appendix

\section{The spectrum of Landau Hamiltonians with bounded potentials}\label{appLsp}

In the appendix we justify  \eq{splandau}.
\begin{lemma}\label{lemmaLsp} Let $H= H_B  +W$, where $H_B$ is the free Landau Hamiltonian as in \eq{free Landau}, and $- M_1\le W\le  M_2$, where $M_1, M_2 \in [0,\infty )$.  Then
\begin{equation} \label{splandau500}
\sigma(H) \subset
 \bigcup_{n=1}^\infty [B_n -  M_1, B_n +M_2].
\end{equation}
\end{lemma}

\begin{proof}
THe lemma follows from  \cite[Theorem~V.4.10]{Ka} by writing
\beq
H=  \pa{H_B -  \tfrac {M_1-M_2}2} + \pa{W+ \tfrac {M_1-M_2}2}.
\eeq
\end{proof}

\section{The spectrum of Anderson-Landau Hamiltonians}  \label{apSp} 

Consider an Anderson-Landau Hamiltonian  $H_{B,\lambda,\omega}=H\up{A}_{B,\lambda,\omega}$ as in \eq{ALH}-\eq{potVL}, and 
suppose that 
\beq  \label{suppmu}
\supp \mu=[-M_1, M_2]
\quad \text{with 
$ M_1,M_2 \in (0,\infty)$} . 
\eeq 
(The argument applies also to the case $ M_1,M_2 \in [0,\infty)$ {with} $ M_1 + M_2 >0$, with the obvious modifications.) In this appendix we make no other hypotheses on the common probability distribution  $\mu$. It follows from  \cite[Theorem 4]{KM2}, which applies also to  Anderson-Landau Hamiltonians, that under these hypotheses we have
\beq\label{KMsp}
\Sigma_{B,\lambda} = \bigcup_{\omega \in \Omega_{\supp}}  \sigma\pa{H_{B,\lambda,\omega}},\quad  \text{where}\quad \Omega_{\supp}:= [-M_1, M_2]^{\Z^2}.
\eeq
We consider squares $\Lambda_L:= [- \frac L 2,  \frac L 2)$  centered at the origin with  side $L>0$. Given such a square $\Lambda$, we define $\omega\up{\Lambda}$ by $\omega\up{\Lambda}_j=\omega_j$ if $j \in \Lambda$ and   $\omega\up{\Lambda}_j=0$ otherwise, and set
\beq \label{ALHfin}
H_{B,\lambda,\omega}\up{\Lambda}:= H_B + \lambda V_\omega\up{\Lambda}, \quad \text{where} \quad V_\omega\up{\Lambda}=V_{\omega\up{\Lambda}}.
\eeq
Note that $ V_\omega\up{\Lambda}$ is relatively compact   with respect to $H_B$, so $\Sigma_B$ is also  the essential spectrum of $H_{B,\lambda,\omega}\up{\Lambda}$.  In particular, $H_{B,\lambda,\omega}\up{\Lambda}$ has discrete spectrum in  the spectral gaps $\set{\mathcal{G}_n(B):=(B_{n},B_{n+1}), \; n=0,1, \dotsc}$ of $H_B$.   Since $\omega\up{\Lambda} \in \Omega_{\supp}$  if $\omega \in \Omega_{\supp}$, it follows that
\beq \label{splandau44}
\Sigma_B\subset \Sigma_{B,\lambda}=\overline{ \bigcup_{n=1}^\infty \bigcup_{\omega \in \Omega_{\supp}} \sigma\pa{H_{B,\lambda,\omega}\up{\Lambda_{L_n}}}},
\eeq
for any $L_n \to \infty$.  (This uses \eq{KMsp} plus the fact that 
$H_{B,\lambda,\omega}\up{\Lambda_{L_n}}$ converges to $H_{B,\lambda,\omega}$ in the strong resolvent sense.)  In particular, it follows from \eq{suppmu}
that  $\Sigma_{B,\lambda}$ is increasing with $\lambda$.

Let $\omega \in \Omega_{\supp}$,  $\omega\up{\Lambda} > 0$, that is,   $\omega_j \ge 0 $  for all $j \in \Lambda$ and  $\sum_{j\in \Lambda} \omega_j >0$. In this case $V_\omega\up{\Lambda}\ge 0$,  and 
\begin{equation} \label{splandau55}
\Sigma_B \subset \sigma\pa{H_{B,\lambda,\omega}\up{\Lambda}}\subset
 \bigcup_{n=1}^\infty [B_n , B_n +\lambda M_2].
\end{equation}

We now  use a modified Birman-Schwinger method,  following \cite[Section~4]{FKmg}. We fix  $n\in \N$ and set
\beq
\mathcal{R}(E)= -\sqrt{ V_\omega\up{\Lambda}} \pa{H_B - E}^{-1}\sqrt{ V_\omega\up{\Lambda}} \quad \text{for} \quad E \in (B_n, B_{n+1}),
\eeq
a compact self-adjoint operator. Let $r^+(E)= \max \sigma\pa{ \mathcal{R}(E)}$.  We claim
\beq \label{Rlim}
\lim_{E \downarrow B_n} r^+(E)=  \infty.
\eeq
To see this, let $\Pi_n= \chi_{\{B_n\}}(H_B)$.  Then
\beq \label{Rlim2}
\mathcal{R}(E)=  \tfrac {1}{E- B_n}\sqrt{ V_\omega\up{\Lambda}} \Pi_n\sqrt{ V_\omega\up{\Lambda}} - \sqrt{ V_\omega\up{\Lambda}}\pa{1- \Pi_n} \pa{H_B - E}^{-1}\sqrt{ V_\omega\up{\Lambda}}.
\eeq
Since 
\beq \norm{\sqrt{ V_\omega\up{\Lambda}}\pa{1- \Pi_n} \pa{H_B - E}^{-1}\sqrt{ V_\omega\up{\Lambda}}} \le \frac {M_2 }{B} \quad\text{for} \quad E \in (B_n, B_n+B) ,
\eeq 
\eq{Rlim} follows if we show that $\sqrt{ V_\omega\up{\Lambda}} \Pi_n\sqrt{ V_\omega\up{\Lambda}}\not=0$.  But otherwise we would conclude that  $\sqrt{ V_\omega\up{\Lambda}} \Pi_n=0$  ($A^*A=0$ implies $A=0$), and, since $V_\omega\up{\Lambda}>0$ in an nonempty open set, we would contradict the unique continuation principle.  Now, using \eq{Rlim}, we conclude, as in \cite[Proposition~4.3]{FKmg}, that
$H_{B,\lambda,\omega}\up{\Lambda}$ has an eigenvalue in  $(B_n , B_n +\lambda M_2]$
for all sufficiently small $\lambda >0$. 

Now, let us replace $\omega$ by $M_2$ in the notation if $\omega_j=M_2$ for all $j$, and consider $H_{B,\lambda,M_2}\up{\Lambda}$. Fix $n \in \N$, and let  $E_+\up{\Lambda}(n,B,\lambda)$ denote the biggest eigenvalue of $H_{B,\lambda,M_2}\up{\Lambda}$ in the open interval $(B_n,B_{n+1})$. We have shown  the existence of $E_+\up{\Lambda}(n,B,\lambda)$ for   small $\lambda>0$. By the argument in \cite[Section~VII.3.2]{Ka}, $E_+\up{\Lambda}(n,B,\lambda)$ then exists for  $\lambda \in (0, \lambda_+\up{\Lambda}(n,B))$, with $ \lambda_+\up{\Lambda}(n,B)>0$,  where it is continuous and increasing in $\lambda$. In view of \eq{splandau55}, we have $\lim_{\lambda \downarrow 0} E_+\up{\Lambda}(n,B,\lambda) = B_n$ and  $ \lambda_+\up{\Lambda}(n,B) \ge  \frac {2B}{M_2}$.  In addition, we must either have $ \lambda_+\up{\Lambda}(n,B)=\infty$ or
$\lim_{\lambda \uparrow \lambda_+\up{\Lambda}(n,B)} E_+\up{\Lambda}(n,B,\lambda) = B_{n+1}$. In the latter case we may thus extend $E_+\up{\Lambda}(n,B,\lambda)$ as an increasing, continuous function for $\lambda \in (0,\infty)$ by setting $E_+\up{\Lambda}(n,B,\lambda)= B_{n+1}$ for $\lambda \ge  \lambda_+\up{\Lambda}(n,B)$.

 A similar argument produces a smallest  eigenvalue $E_-\up{\Lambda}(n,B,\lambda) \in [B_{n-1} ,B_n)$ of $H_{B,\lambda,-M_1}\up{\Lambda}$ in $(B_{n-1},B_{n})$
for $\lambda \in (0, \lambda_-\up{\Lambda}(n,B))$, where $ \lambda_-\up{\Lambda}(n,B) \ge  \frac {2B}{M_1}$,   continuous and decreasing in $\lambda$,  with 
$\lim_{\lambda \downarrow 0} E_-\up{\Lambda}(n,B,\lambda) = B_n$.  Moreover, $ \lambda_-\up{\Lambda}(1,B)=\infty$, and, for $n=2,3,\ldots$,  either  $ \lambda_-\up{\Lambda}(n,B)=\infty$ or
$\lim_{\lambda \uparrow \lambda_-\up{\Lambda}(n,B)} E_-\up{\Lambda}(n,B,\lambda) = B_{n-1}$. In the latter case we extend $E_-\up{\Lambda}(n,B,\lambda)$ as a decreasing, continuous function for $\lambda \in (0,\infty)$ by setting $E_-\up{\Lambda}(n,B,\lambda)= B_{n-1}$ for $\lambda \ge  \lambda_-\up{\Lambda}(n,B)$.

 For an arbitrary  $\omega \in \Omega_{\supp}$ and $\lambda>0$, the eigenvalues of  $H_{B,\lambda,\omega}\up{\Lambda}$ in the intervals $(B_n, B_n + \lambda M_2)$ and  $(B_n -\lambda M_1, B_n)$ (if they exist) are separately continuous and increasing in each $\omega_j\in [-M_1,M_2]$, $j \in \Lambda$, and hence they must be in the interval $ I_n\up{\Lambda}(B,\lambda)= [E_-\up{\Lambda}(n,B,\lambda), E_+\up{\Lambda}(n,B,\lambda)]$. Thus we conclude that for each
square $\Lambda$ we have
\beq \label{splandau66}
\bigcup_{\omega \in \Omega_{\supp}} \sigma\pa{H_{B,\lambda,\omega}\up{\Lambda}}
= \bigcup_{n\in \N} I_n\up{\Lambda}(B,\lambda).
\eeq
In addition, the same argument shows that for fixed $\lambda$ and $B$ we have  $\pm E_\pm\up{\Lambda}(n,B,\lambda) $   increasing with $\Lambda$.
We  set $ E_+ (n,B,\lambda):= \sup_{\Lambda} E_+ \up{\Lambda}(n,B,\lambda)\le B_{n + 1}  $,  $ E_-  (n,B,\lambda):= \inf_{\Lambda} E_- \up{\Lambda}(n,B,\lambda)\ge B_{n-1}  $, and conclude from \eq{splandau44} and \eq{splandau66} that
(cf. \cite[Eq.~(2.11)]{GKS}
\beq \label{splandau667}
\Sigma_{B,\lambda} = \bigcup_{n\in \N} I_n(B,\lambda), \quad \text{where} \quad  I_n(B,\lambda)= [E_-(n,B,\lambda), E_+(n,B,\lambda)].
\eeq
Note that the intervals $ I_n(B,\lambda)$ depend on $\supp \mu=[-M_1,M_2]$, but not on other details of the measure $\mu$.

Now assume that $u$ in \eq{potVL} satisfies 
\beq\label{covcond5}
0< U_- \le U(x):= \sum_{i \in\mathbb{Z}^2}  u(x-i) \le  1,
\eeq
for some constant $U_-$. (The upper bound is simply a normalization we had already assumed.)  In this case,  for all $n \in \N$ we have 
\begin{align}\label{locsp}
B_n + \lambda M_2 U_- & \le E_+(n,B,\lambda)  \quad \text{for} \quad \lambda \in \pa{0, \tfrac {2B}{M_2 U_-}},\\
B_n - \lambda M_1 U_-  & \ge E_-(n,B,\lambda) \quad \text{for} \quad \lambda \in \pa{0, \tfrac {2B}{M_1 U_-}}.\label{locsp25}
\end{align}
We also have
\beq \label{locsp256}
B - \lambda  M_1 U_-    \ge E_-(1,B,\lambda) \quad \text{for all} \quad \lambda \ge 0.
\eeq
This can be seen as follows.  Take $ \lambda \in(0, \frac {2B}{M_2 U_-})$, then
\beq
H_{B,\lambda,M_2}=  H_B + \lambda M_2 U_- +  \lambda M_2 (U-U_-), \;\; \text{with} \;\;0\le U-U_-\le 1- U_-.
\eeq
Since $\sigma\pa{ H_B + \lambda M_2 U_-}=
\Sigma_B + \lambda M_2U_-=\set{B_n +\lambda M_2 U_- ;\  n\in \N }  $, it follows from  \cite[Theorem~4.10]{Ka}  (as in Lemma~\ref{lemmaLsp}), and the definition of $ E_+(n,B,\lambda)$, that
\begin{equation} \label{splandau66666}
\sigma\pa{H_{B,\lambda,M_2}}\subset
 \bigcup_{n=1}^\infty 
[B_n +\lambda M_2 U_- ,  E_+(n,B,\lambda)] . 
\end{equation}
Since by the same argument
\beq\label{samearg}
\Sigma_B + \lambda M_2U_-\subset   \bigcup_{n\in \N_{\not=\emptyset}  }
[B_n +\lambda M_2 U_- -\lambda M_2(1- U_-)  , E_+(n,B,\lambda) ] ,
\eeq
where $ \N_{\not=\emptyset} :=\set{n\in \N;\; \sigma\pa{H_{B,\lambda,M_2}} \cap [B_n +\lambda M_2 U_- ,  E_+(n,B,\lambda)] \not= \emptyset }$,
we conclude that   $ \N_{\not=\emptyset}=\N$. 
It then follows from \eq{splandau667} that \eq{locsp} holds.  \eq{locsp25} and  \eq{locsp256} are proved in a similar manner.

Under the condition \eq{gapcond} the spectral gaps never close.  On the other hand, if 
we have \eq{covcond5},  if
\beq \label{gapclose}
\lambda U_- (M_1 + M_2) \ge 2B,
\eeq
all the internal spectral gaps close, i.e.,
\beq
\Sigma_{B,\lambda} = \pa{E_- (1,B,\lambda), \infty}.
\eeq

\end{document}